\pdfoutput=1
\documentclass[reqno,oneside,letterpaper, 12pt]{amsart}
\usepackage{mathrsfs}
\usepackage[foot]{amsaddr}
\usepackage{amssymb}
\usepackage{wasysym}
\usepackage[utf8]{inputenc}
\usepackage{amsmath, amssymb,bm, cases, mathtools, thmtools}
\usepackage{verbatim}
\usepackage{graphicx}\graphicspath{{figures/}}
\usepackage{multicol}
\usepackage{tabularx}
\usepackage[usenames,dvipsnames]{xcolor}
\usepackage{mathrsfs}
\usepackage{url}
\usepackage[normalem]{ulem}
\usepackage{xstring}
\usepackage[shortlabels]{enumitem}

\usepackage[%
    minnames=1,maxnames=99,maxcitenames=3,
    style=authoryear,
    doi=false,url=false,
    firstinits=true,hyperref,natbib,backend=bibtex]{biblatex}
\renewbibmacro{in:}{%
  \ifentrytype{article}{}{\printtext{\bibstring{in}\intitlepunct}}}
\bibliography{biblio}

\usepackage[colorlinks,citecolor=blue,urlcolor=blue,linkcolor=RawSienna]{hyperref}
\usepackage{hypernat}
\usepackage{datetime}

\DeclareMathAlphabet\EuRoman{U}{eur}{m}{n}
\SetMathAlphabet\EuRoman{bold}{U}{eur}{b}{n}

\usepackage{euscript,microtype}
\usepackage[capitalize]{cleveref}

\crefname{assumption}{Assumption}{Assumptions}
\crefname{claim}{Claim}{Claims}

\makeatletter
\let\reftagform@=\tagform@
\def\tagform@#1{\maketag@@@{\ignorespaces\textcolor{gray}{(#1)}\unskip\@@italiccorr}}
\renewcommand{\eqref}[1]{\textup{\reftagform@{\ref{#1}}}}
\makeatother

\definecolor{WowColor}{rgb}{.75,0,.75}
\definecolor{SubtleColor}{rgb}{0,0,.50}

\newcounter{margincounter}

\declaretheorem[style=plain,numberwithin=section,name=Theorem]{theorem}
\declaretheorem[style=plain,sibling=theorem,name=Lemma]{lemma}

\declaretheorem[style=plain,sibling=theorem,name=Claim]{claim}

\declaretheorem[style=definition,sibling=theorem,name=Definition]{definition}

\declaretheorem[style=definition,sibling=theorem,name=Example]{example}
\declaretheorem[style=remark,sibling=theorem,name=Remark]{remark}
\crefformat{conditionINNER}{#2(#1)#3}
\crefmultiformat{conditionINNER}
  {(#2#1#3)}
  { and~(#2#1#3)}
  {, (#2#1#3)}
  { and~(#2#1#3)}
\Crefformat{conditionINNER}{Condition~#2(#1)#3}
\Crefmultiformat{conditionINNER}
  {Conditions~(#2#1#3)}
  { and~(#2#1#3)}
  {, (#2#1#3)}
  { and~(#2#1#3)}

\declaretheoremstyle[
    spaceabove=-6pt,
    spacebelow=6pt,
    headfont=\normalfont\bfseries,
    bodyfont = \normalfont,
    postheadspace=1em,
    qed=$\square$,
    headpunct={{}}]{myproofstyle}

\numberwithin{equation}{section}
\numberwithin{theorem}{section}

\renewcommand{\appendixname}{}
\usepackage{amssymb}
\usepackage{mathrsfs}
\usepackage{leftidx}

\usepackage{setspace}
\usepackage{thmtools}
\usepackage{geometry}
\usepackage{enumitem}
\usepackage{caption}
\usepackage[colorinlistoftodos]{todonotes}
\def\[#1\]{\begin{align}#1\end{align}}
\def\*[#1\]{\begin{align*}#1\end{align*}}

\newcommand{\Reals}{\mathbb{R}}

\newcommand{\NNReals}{\Reals_{\ge 0}}

\newcommand{\conv}{\textrm{conv}}

\DeclareMathOperator*{\newlim}{\mathrm{lim}\vphantom{\mathrm{infsup}}}

\DeclareMathOperator*{\newsup}{\mathrm{sup}\vphantom{\mathrm{infsup}}}
\renewcommand{\lim}{\newlim}

\renewcommand{\sup}{\newsup}

\newcommand{\PowerSet}{\mathscr{P}}

\newcommand{\cA}{\mathcal{A}}

\newtheorem{open problem}{Open Problem}

\newcommand{\interior}[1]{%
  {\kern0pt#1}^{\mathrm{o}}%
}

\newcommand{\refproof}[1]{See \cref{#1} for \IfSubStr{#1}{,}{proofs}{a proof}. }

\newif\iflongform
\longformtrue
\iflongform

\else

\fi

\renewcommand\thmcontinues[1]{Continued}

\makeatletter
\providecommand*{\toclevel@definition}{0}
\providecommand*{\toclevel@theorem}{0}
\providecommand*{\toclevel@lemma}{0}
\makeatother
\usepackage{leftidx}

\makeatletter
\patchcmd{\@maketitle}
  {\ifx\@empty\@dedicatory}
  {\ifx\@empty\@date \else {\vskip3ex \centering\footnotesize\@date\par\vskip1ex}\fi
   \ifx\@empty\@dedicatory}
  {}{}
\patchcmd{\@maketitle}   % <----- This line changed
  {\ifx\@empty\@date\else \@footnotetext{\@setdate}\fi}
  {}{}{}
\makeatother

\newtheorem{innercustomthm}{Theorem}
\newenvironment{customthm}[1]
{\renewcommand\theinnercustomthm{#1}\innercustomthm}
{\endinnercustomthm}

\newtheorem{innercustomexp}{Example}
\newenvironment{customexp}[1]
{\renewcommand\theinnercustomexp{#1}\innercustomexp}
{\endinnercustomexp}

\newtheorem{innercustomcor}{Corollary}
\newenvironment{customcor}[1]
{\renewcommand\theinnercustomcor{#1}\innercustomcor}
{\endinnercustomcor}

\title[]{General Equilibrium Theory for Climate Change}

\author[R.~M.~Anderson]{Robert M.~Anderson $^{1}$}
\address{$^{1}$ University of California, Berkeley, Department of Economics}
\email{robert.anderson@berkeley.edu}
 
\author[H.~Duanmu]{Haosui Duanmu $^{2}$}
\address{$^{2}$ Institute for Advanced Study in Mathematics, Harbin Institute of Technology}
\email{duanmuhaosui@hotmail.com}

\date{September 8, 2023}

\newcommand{\cto}{\twoheadrightarrow}

\newcommand{\cO}{\mathcal{O}}

\newcommand{\NegReals}{\mathbb{R}_{\leq 0}}

\newcommand*{\argmax}{\operatornamewithlimits{arg max}\limits}

\begin{document}

\thanks{This paper grew out of work, joint with M. Ali Khan and Metin Uyanik, concerning externalities in production economies with a continuum of consumers. We are grateful to John Geanakoplos, Aniruddha Ghosh, Felix Kubler, Michael Mandler, Herakles Polemarchakis, Max Stinchcombe, and to participants in a seminar at Zhejiang University, the NSF/CEME Decentralization Conference at the University of Michigan, the 17th Annual Cowles Conference on General Equilibrium and its Applications at Yale, a seminar at Macau University of Science and Technology, the William R. Zame Conference at UCLA, 22nd annual SAET Conference in Paris, and the 2023 Asian Meeting of the Econometric Society at Nanyang University of Technology; for comments that greatly improved our results. We are also grateful to Tian Luo for carefully reading the paper, and to Yahui Zhang for producing the figures in the manuscript.  Anderson gratefully acknowledges financial support from Swiss Re through the Consortium for Data Analytics in Risk.}

%\author[M.~A.~Khan]{M.~Ali Khan}
%\address{M.~Ali Khan, The Johns Hopkins University, Department of Economics}

%\author[W.~Weiss]{William Weiss}
%\address{William Weiss, University of Toronto, Department of Mathematics}

%\author[A.~Smith]{Aaron Smith}
%\address{University of Ottawa, Department of Mathematics and Statistics}

\maketitle

\begin{abstract}
%We propose two general equilibrium (GE) models---quota equilibrium and emission tax equilibrium---that incorporate the regulatory schemes used in practice to control greenhouse gas emissions. In these models, the government first specifies quotas or taxes on emissions, then refrains from further action. We show the existence of a quota equilibrium. In a quota equilibrium, the allocation of emission property rights has a major effect on the distribution of welfare among consumers.
%Assuming that the only externality arises from the total net pollution emission, the quota equilibrium consumption-production plan is Pareto Optimal among all feasible consumption-production plans with the same total net emissions. We show that every quota equilibrium can be realized as an emission tax equilibrium and vice versa. However, for certain tax rates, there may be no corresponding emission tax equilibrium. 
%Full Pareto optimality of quota equilibrium can often be achieved by setting the right quota, and, in some cases, full Pareto optimality of emission tax equilibrium can be achieved by setting the right tax rate. 

We propose two general equilibrium models, quota equilibrium and emission tax equilibrium. The government specifies quotas or taxes on emissions, then refrains from further action. Quota equilibrium exists; the allocation of emission property rights strongly impacts the distribution of welfare. If the only externality arises from total net emissions, quota equilibrium is constrained Pareto Optimal.  Every quota equilibrium can be realized as an emission tax equilibrium and vice versa. However, for certain tax rates, emission tax equilibrium may not exist, or may exhibit high multiplicity. Full Pareto Optimality of quota equilibrium can often be achieved by setting the right quota.

\end{abstract}

\section{Introduction}

The mitigation of climate change requires the reduction of greenhouse gas emissions.  In the absence of governmental regulation, market mechanisms have proven insufficient to achieve the necessary reduction.  Various regulatory schemes---notably ``carbon taxes'' and cap and trade---have been proposed. However, these regulatory schemes have yet to be incorporated into a general equilibrium model of the generality and rigor of the Arrow-Debreu Model \citep{ad54}.  In this paper, we do so.

\subsection{Regulatory Schemes Used in Practice}
Three principal schemes are used in practice to regulate greenhouse gas emissions.  In each case, the government first selects certain commodities whose emissions will be regulated.\footnote{For example, carbon dioxide $\mathrm{CO}_{2}$, methane $\mathrm{CH}_4$, and chlorofluorocarbons (CFCs).}  ``Cap and trade'' is a scheme in which each private firm is assigned an emissions quota; this quota is a property right, namely the right to emit regulated commodities up to the quota.  A firm may exceed its emissions quota, but only provided that it buys the right to do so from another firm whose emissions fall short of its quota. Our notion of quota equilibrium generalizes cap and trade by allowing the government to assign itself a quota. Quota equilibrium thus includes two important polar cases: cap and trade, in which all the quota is assigned to private firms, and global quota equilibrium, in which all of the quota is assigned to the government and none is assigned to private firms. 

An alternative to cap-and-trade is a {\em fixed quota} model, in which the government prohibits each private firm from exceeding its emissions quota.  Fixed quotas are used in the regulation of local forms of pollution, which primarily affect a small geographic region. For example, fixed quotas limit the discharge of toxic chemicals from a given plant into ground water immediately adjacent to that plant.

An {\em emissions tax} (commonly called a ``carbon tax'') is a tax per unit of emissions of a regulated commodity.\footnote{The term ``carbon tax'' is somewhat imprecise, and we use the term emissions tax throughout.}  A {\em fuel tax} is a tax $t$ per unit of a fuel, such as gasoline. Note that an emissions tax is {\em the} price of emitting the commodity into the atmosphere; the Walrasian auctioneer has no flexiblity to alter that price in order to equilibrate supply and demand. By contrast, a fuel tax is an add-on tax; if a consumer $\omega$ buys a gallon of gasoline from firm $j$, $\omega$ pays a price $p+t$; the firm remits $t$ to the government, and retains revenue $p$, so there is a wedge of $t$ between the price the consumer pays and the price the firm receives. Even though $t$ is fixed, market forces can cause $p$ to vary.

\subsection{Intuition from Partial Equilibrium Analysis}

Cap and trade, emissions taxes, and fuel taxes, have been extensively studied in {\em partial equilibrium}.\footnote{See, e.g. \citet{Barnett80}, \citet{fischer03}, \citet{Ciaian06} and \citet{Fla09}.
In addition, \citet{sandmo75} points out that ``most of the literature on optimal taxation is partial equilibrium in nature and does not contain satisfactory treatment of the benefits resulting from control of externalities."}  
It is important to note that, unlike local pollution such as particulates, greenhouse gases dissipate throughout the atmosphere, and are quite persistent. Hence, the amount of warming depends only on the integral of past emissions (discounted only slightly for older emissions, because of the persistence), {\em regardless of where they were emitted.} 
Given this, it is intuitive in partial equilibrium that cap and trade and emissions taxes are superior to fixed quotas for the regulation of greenhouse gases, for two reasons:
\begin{itemize}
    \item Cap and trade and emission taxes provide incentives for emission reductions to occur with the least reduction in economic output. 
    \item Emissions taxes and cap and trade can suffice to achieve the same emission reductions, but with different distributional consequences.\footnote{Because the property rights allocated to firms in cap and trade increase firm profits, and thereby benefit shareholders (typically wealthier than non-shareholders), while emission taxes can be rebated equally to all consumers, or disproportionately to low income consumers, it is argued that emissions taxes should result in less inequality than cap and trade.}
\end{itemize}

However, these partial equilibrium arguments do not have a rigorous foundation, for three fundamental reasons:
\begin{enumerate}
    \item Once the government sets the quotas or taxes, and the rebate scheme, partial equilibrium arguments cannot guarantee that an equilibrium price exists. This concern is underscored by \cref{taxequipbexample} and \cref{addonnoexst} showing that emission tax and fuel tax equilibria in fact {\em need not exist.}  This is almost obvious in the case of emissions tax equilibrium. The emissions tax rate set by the government freezes the price of emissions, so there are more commodity markets for the Walrasian auctioneer to equilibrate than there are prices  that the Walrasian auctioneer can adjust.  In the case of fuel tax equilibrium, the fixed tax is added to the variable price of the commodity, so the number of commodity markets equals the number of prices that the Walrasian auctioneer can adjust, but {\em equilibrium may still fail to exist.}  
    \item Formulating the notion of Pareto Optimality requires the definition of the set of all feasible consumption-production pairs, but these cannot be defined in partial equilibrium.  Thus, it is impossible to rigorously establish welfare results through partial equilibrium arguments.
    \item Introducing a quota or tax, even on a limited set of commodities, can result in substantial changes throughout the economy, including in sectors that appear to have little to do with the regulated commodities.  These changes are not predictable through partial equilibrium arguments. This concern is underscored by \cref{taxequipbexample} showing that, due to multiplicity of equilibria, it {\em may be impossible to set an emissions tax rate so that all emissions equilibria result in emissions falling below a given target.}
\end{enumerate}
One of the major contributions of this paper is to make the intuitions from partial equilibrium theory rigorous and precise by embedding cap and trade and emissions taxes in a general equilibrium framework, at the rigor and generality of the Arrow-Debreu Model.

\subsection{Pigouvian Taxation}
In the presence of externalities, the social cost of a production or consumption activity does not equal the private cost faced by a firm or agent, and hence the equilibrium need not satisfy the first order conditions for Pareto Optimality.  \citet{pigou20} proposed levying an add-on tax equal to the difference between the social and private costs, in order to restore the first order conditions. There is a considerable literature on Pigouvian taxation through partial equilibrium analysis. We simply note that the act of levying a tax results in a new equilibrium with different prices, productions and consumptions, and these changes occur throughout the economy, not just in the sector(s) in which the tax is levied.  At the new equilibrium, the marginal costs and benefits are different, and while the partial equilibrium analysis suggests that the tax ought to improve outcomes, there is no reason to think that the same tax will restore the first order conditions for Pareto optimality at the new equilibrium.\footnote{\citet{Goudler03} point out that there is a substantial bias from ignoring general equilibrium effects in estimating excess burden.}

\citet{sandmo75} initiates a line of research on optimal taxation with externalities through general equilibrium analysis. Subsequent literature applies \citet{sandmo75}'s framework to environmental policy, see e.g. \citet{bovenberg94jpe}, \citet{bovenberg94aer}, \citet{bovenberg96}, \citet{Goudler03}, \citet{metcalf03}, \citet{Golosov14}, \citet{Goudler16}, and \citet{Goudler19}. \citet{sandmo75} and the subsequent authors\footnote{\citet{Goudler19} depart from Sandmo's framework by combining two models. The first is a representative agent model that is used to determine the prices of commodities, both final consumption goods and the inputs to production. The second model has heterogeneous agents (5 in the numerical simulations). Meshing these two models allows them to estimate the effects of a carbon tax policy on different income quintiles.} consider a representative agent model with a single producer with linear technology; 
\begin{enumerate}

\item The first-order conditions for Pareto optimal allocations and competitive equilibria are computed separately through solving the Lagrangian.\footnote{Equating the first-order conditions tacitly assumes that the representative agent's equilibrium consumption is strictly positive.} \citet{sandmo75} then equates these first-order conditions to determine when a competitive equilibrium allocation is Pareto optimal. The resulting tax rate is called the \emph{optimal Pigouvian tax rate}.
However, due to multiplicity of equilibria, setting the tax rate equal to the optimal Pigouvian tax rate does not necessarily result in Pareto optimal equilibrium,
as shown in \cref{taxequipbexample}.

\item \citet{sandmo75} and subsequent authors do not prove the existence of equilibrium at the optimal Pigouvian tax rate. They do not study the welfare properties of equilibria with varying tax rates. 
\end{enumerate}

\subsection{Comparison of our GE Analysis to Partial Equilibrium and to Pigouvian Taxation}
We embed quota, emissions and fuel taxes into a GE model of the generality and rigor of the Arrow-Debreu model.
Our GE model allows for multiple agents with possibly different consumption sets, preferences, endowments and shareholdings of firms. Moreover, it allows for multiple firms with possibly nonlinear production technology.  
Thus, \citet{sandmo75}'s representative agent model is a special case of our model.

In our model, the government may generate revenue by selling the quota it assigns to itself in the quota model, while the government generates tax revenue in the emissions and fuel tax models. 
In all three models, the government sets a rebate scheme that specifies how its revenues will be distributed to consumers. The government first specifies a quota (along with its allocation between the private firms and the government firm) or a tax rate, and a rebate scheme. At that point, it abstains from further action, allowing market forces (i.e., the Walrasian auctioneer) to determine a price that equilibrates supply and demand. Here is a comparison of our model to the partial equilibrium formulation, and to Pigouvian taxation:
\begin{enumerate}
    \item We show in \cref{standardqtmain} that, under mild assumptions, given any quota that is consistent with a consumption-production pair, there is an associated quota equilibrium at which the total net pollution emission equals the given quota. The existence is proven by shifting firms' production sets by their prespecified quotas, invoking Proposition 3.2.3 in \citet{fl03bk}, and showing that the equilibrium in the derived Florenzano's economy is a quota equilibrium in our model. Quota equilibrium allows considerable flexibility in the allocation of emission property rights, ranging from cap and trade equilibrium, in which all emission property rights are vested in private firms, to global quota equilibrium, in which all emission property rights are vested in the government. \cref{exampleproperty} demonstrates that the quota property rights have a major impact on the distribution of welfare among agents;
    
    \item The desire to regulate pollution emission is driven by an important externality, that a given amount of total emissions results in a corresponding increase in the average global temperature.  If there are no other externalities in the economy, then a version of the First Welfare Theorem holds: taking the specified emissions as given, then for any distribution scheme, the resulting equilibrium consumption-production pair is constrained Pareto optimal, i.e., Pareto Optimal among the set of all quota-compliant consumption-production pairs with the same total net emissions of regulated commodities, see \cref{fstwelfarequota} and \cref{welfareprod}. For a fixed distribution scheme, changing the quota alters the consumers' welfare, and it is possible that the equilibrium consumption-production pair for one quota may Pareto dominate the equilibrium consumption-production plan for a different quota, see \cref{quotawelfareexample}. But once the government has established the quota, no further government intervention is required to achieve constrained Pareto optimality.  In particular, changing the rebate scheme may benefit some consumers and disadvantage others, but it cannot result in a Pareto improvement.

    Although our terminology, discussion and examples in this paper are heavily focused on $\mathrm{CO}_{2}$ emissions and climate change, our model allows for very general externalities in agents' preferences. In particular, our main existence result (\cref{standardqtmain}) applies to a wide variety of externalities, including the regulation of other forms of pollution. By essentially the same proof, a generalization of \cref{fstwelfarequota} shows that every quota equilibrium consumption-production pair is Pareto optimal among all quota-compliant consumption-productions pairs that generate the same externality;\footnote{In the case of local pollution, such as pollution discharge into a river, each location would have its own quota. ``Generating the same externality" requires having the same emissions at every local site.}

    \item In the quota equilibrium model, the government refrains from further intervention after setting the quota (along with its allocation) and the shareholdings of the government firm. Thus, the equilibrium price for the quota is determined through market forces. Similarly, the equilibrium total net emissions of regulated commodities in the emission tax equilibrium model is determined endogenously through market forces, after the government has set the tax rate and its rebate share to agents;
    
    \item We show under mild assumptions that 
    every quota equilibrium is an emission tax equilibrium for a carefully chosen government rebate scheme that captures the property rights assignment embedded in the quota.     
    Moreover, every emissions tax equilibrium is a global quota equilibrium, i.e. the equilibrium of an economy in which all of the quotas are assigned to the government, and none to the private firms.  This makes rigorous an assertion of \citet{Nordhaus79} on the interchangeability of carbon taxes and emissions quotas.\footnote{Chapter 8 of \citet{Nordhaus79} is devoted to a pioneering study of strategies for controlling $\mathrm{CO}_{2}$ emissions. Nordhaus begins with ``an efficient path'' of emissions and seeks a policy to provide agents the correct incentives to achieve it. Nordhaus writes, ``In the real world, the policy can take the form either of taxing carbon emissions, or of physical controls (such as rationing).  In an efficient solution, the two are interchangeable in principle''[page 136].  It is not clear whether Nordhaus, in this statement, contemplated trading of the emission rations.}  In other words, the distributional consequences arise not from the choice of quota (quantity) versus tax (price), but rather from the way property rights are distributed between private firms and the government in the quota setting, and the way government revenues are rebated to consumers;
    
    \item As noted above, \cref{taxequipbexample} show that emissions tax equilibrium may fail to exist for certain tax rates. \cref{taxequipbexample} further shows that a given emissions tax rate may be associated with multiple emissions tax equilibria with different total net emissions of regulated commodities, hence {\em it may not be possible to set an emissions tax rate so that all associated emissions equilibria result in total net emissions of regulated commodities falling below a given target.}\footnote{The multiplicity of emissions tax equilibrium resolves what at first might appear to be a conflict between this bullet and the first two bullets.}

    \item Full Pareto optimality of quota equilibrium can often be achieved if the government sets the right quota. In addition, \cref{nonlinearexample} suggests that, when there is a one-to-one correspondence between the emission tax rate and total net pollution emission, the government can often limit the total net emissions of regulated commodities to a given level and achieve full Pareto optimality, through setting the tax rate;

    \item We focus on emissions tax equilibrium rather than fuel tax equilibrium. In \cref{cbtaxequil}, we find that fuel tax equilibrium may be inefficient compared with emissions tax equilibrium, i.e., for a given total net $\mathrm{CO}_{2}$ emission, the corresponding fuel tax equilibrium is Pareto dominated by the corresponding emissions tax equilibrium. \cref{addonnoexst} shows that fuel tax need not exist for certain tax rates;

    \item  Computable General Equilibrium (CGE) methods\footnote{\citet{scarf72} demonstrated the feasibility of computing Walrasian equillibrium by providing an algorithm guaranteed to converge to an equilibrium.  Since then, CGE methods have become widely used.  In practice, CGE methods typically use some variation of Newton's Method, which may fail to converge, but converges rapidly when it does.  When convergence fails, the algorithm can be rerun from a different starting point until convergence is hopefully achieved.} allow for the effective computation of our equilibrium notions in country macroeconomic models. This allows the government to predict the equilibria that will arise from a variety of possible regulatory policies, facilitating the government's choice process.  A particular change that might be beneficial for a vast majority of agents is nonetheless likely to be detrimental for some agents.  Moreover, these agents may be engaged in economic activities that, at first, seem only remotely connected to the regulated sector. If the government can identify those losers in advance of implementing the policy change, it may be able to mitigate those losses by other actions; for example, agents that will lose jobs may be offered subsidized job retraining, allowing them to switch to expanding sectors of the economy.  The ability to mitigate harmful side effects of a generally beneficial regulation is of value in its own right, and may be of particular value in assembling a political consensus in favor of implementing the change. 
    CGE methods have been widely used in trade negotiations, and should be of value in negotiations to set the emissions quotas or taxes in individual countries.\footnote{See, for example, \citet{shoven72}, \citet{felix04}, and \citet{felix06}.}  
\end{enumerate}

\subsection{Comparison with Other Approaches}
We briefly discuss three other possible approaches to climate regulation, and how our model compares to these approaches.

\subsubsection{Disposal Cone}
Free-disposal equilibrium and non-free-disposal equilibrium are two classical equilibrium notions, which differ in the resource feasibility constraint. The demand is allowed to be less than or equal to the supply for free-disposal equilibrium, while non-free-disposal equilibrium notion requires demand to be exactly equal to the supply for each commodity.\footnote{While \citet{ad54} used free-disposal equilibrium as the equilibrium concept, \citet{McKenzie81}
used non-free-disposal equilibrium as the equilibrium concept. See also \citet{Berg71}, \citet{Hart75} and \citet{pole93}.} As pointed out in \citet{ha05}, the free-disposal equilibrium notion allows bads to be freely disposed and ``hence trivializes the problem of the efficient allocation of bads." \citet{fl03bk} formulated the notion of a disposal cone, which allows the exogenous specification of certain commodities that can be freely disposed, and other goods that may not be disposed. 
%In addition, Florenzano does not assume free-disposal in production, through which it is possible to disguise a free-disposal equilibrium as a non-free-disposal equilibrium. 

Florenzano's work is a substantial conceptual improvement on the pre-existing literature, but it has two significant limitations: 
\begin{enumerate}
\item While one can prohibit the emissions of a given pollutant, the cone formulation does not allow for setting a positive cap on those emissions. In \citet{fl03bk}, a ``complementary slackness'' condition holds: the value of disposal at equilibrium of any good must be zero. So it is impossible to charge a positive tax on pollution. The complementary slackness condition is a consequence of budget balance in the Arrow-Debreu model. If a tax on pollution generates positive revenues, these revenues would evaporate from the model, leaving consumers without sufficient income to buy the goods produced by the firms;
\item \citet{fl03bk} establishes the existence of equilibrium such that the equilibrium excess demand is in her disposal cone. However, every non-free-disposal equilibrium is also a Florenzano equilibrium, since her equilibrium conditions are less stringent. It would be desirable to establish the existence of equilibria that are not non-free-disposal, but Florenzano does not do so. Since it is not practical to eliminate all pollution in the near term, we need to demonstrate the existence of equilibrium with non-zero disposal of bads.  
\end{enumerate}
Our quota equilibrium model, on the other hand, incorporates a quota regulatory scheme by defining the quota-compliance region $\mathcal{Z}(m)$ that reflects the society's choice on which commodities to dispose, and at what quantity. Our quota equilibrium notion requires that the total net emissions of the regulated commodities equals  some given quota. Any revenue generated from quota selling or tax is rebated to agents according to either the agents shareholdings of private firms or government's rebate scheme to ensure that all revenues remain within the model. 

\subsubsection{Lindahl Equilibrium}
In an economy with a public good, a Lindahl equilibrium (see \citet{foley70})
is a system of individualized prices charged per unit of public good, such that all agents agree on a common level of the public good.  Lindahl equilibrium is Pareto efficient.  However, computing the individualized prices at Lindahl equilibrium requires knowing the demand functions of all agents. Moreover, individuals have an incentive to misrepresent their demand functions in order to lower their share of the cost of providing the public good.  Thus, there are serious difficulties in implementing Lindahl equilibrium in practice.

In contrast to Lindahl equilibrium, cap and trade and emission taxes can be---indeed, have been---implemented in practice.  Our goal is to understand the consequences of implementing {\em those} regulatory schemes, in particular the existence or nonexistence, constrained Pareto optimality, and potential full Pareto optimality of the equilibria that arise.  In both quota and emissions tax equilibrium, the prices of emissions are uniform across payers, rather than individualized as in Lindahl equilibrium. Thus, Lindahl equilibrium provides little guidance on the behavior of the regulatory schemes---cap and trade and emissions tax---that have been implemented in practice.

\subsubsection{Fixed Price Equilibrium}

\citet{dreze75} considers exchange economies in which the prices (or relative prices) of certain commodities are subject to inequality constraints, so the Walrasian auctioneer may not have enough freedom in varying prices to equilibrate supply and demand. Dr\`eze defines a Fixed Price Equilibrium,\footnote{Dr\`eze does not use the term ``Fixed Price Equilibrium,'' but his equilibrium notion has come to be called that in the subsequent literature, for example see \citet{Guy78}.} in which rationing is used to equilibrate supply and demand, at the fixed prices. 
In contrast to our model, Dr\`eze does not allow for externalities, which are the essential motivation for regulating pollution emission. 
In addition, a Fixed Price Equilibrium sets constraints on both prices {\em and} quantities in a single equilibrium; we set either quantities (quotas) or prices (taxes), but we never set both simultaneously. Our notions capture the regulatory schemes (cap and trade and carbon tax) that are in actual use to control climate change. Note that both of these regulatory schemes set quantities or prices, but not both, so there is no apparent connection between Fixed Price equilibrium, on the one hand, and our equilibrium notions or the regulatory schemes used in practice, on the other hand.

\subsection{Structure of the Paper}

Here is an outline of the paper. \cref{secmodel} furnishes the rigorous framework of the quota equilibrium model. The existence of quota equilibrium is presented in \cref{secquotaeq} (see \cref{standardqtmain}), and the constrained Pareto optimality of quota equilibrium is presented in \cref{secwelfare} (see \cref{fstwelfarequota}). Furthermore, \cref{exampleproperty} illustrates that quota allocation may have a major impact on the distribution of welfare among agents. \cref{sectaxequil} presents the emissions tax equilibrium model, and compares the emissions tax equilibrium with quota equilibrium. In particular, every quota equilibrium can be viewed as an emissions tax equilibrium and vice versa (see \cref{equndertaxisquota} and \cref{equnderqtistax}). However, \cref{taxequipbexample} shows that emissions tax equilibrium may not exist for certain tax rates, and may be ineffective in limiting the total net emission of regulated commodities. On the other hand, \cref{nonlinearexample} suggests that, when there is a one-to-one correspondence between emission tax rate and total net pollution emission, the government can often limit the total net emissions of regulated commodities at a given level and achieve full Pareto optimality, through tax alone. \cref{secfuel} presents two examples on fuel tax equilibrium; \cref{cbtaxequil} suggests that fuel tax equilibrium may be Pareto dominated by the corresponding emissions tax equilibrium, and \cref{addonnoexst} shows that fuel tax equilibrium need not exist for certain tax rates. \cref{secconclude} concludes the paper and lays out a few promising future directions. 
Finally, \cref{secappa} completes the proof of our main existence result, \cref{standardqtmain}. \cref{finitemainexpd} transforms  Proposition 3.2.3 in \citet{fl03bk} into our setting. The derivation of \cref{standardqtmain} from \cref{finitemainexpd} is given in \cref{secquotaeq}.

\section{Production Economy with Quota}\label{secmodel}

In this section, we present a general equilibrium model that incorporates quota regulatory schemes on net pollution emission, and introduce an equilibrium notion known as quota equilibrium. In particular, we generalize the usual feasibility constraints so that, at equilibrium, the total net pollution emission equals the prespecified level. Our model incorporates two important polar cases: cap and trade equilibrium, in which the emission quota is allocated entirely to private firms; and global quota equilibrium, in which the emission quota is vested in a government firm. 

Most of our results apply regardless of how the emission quota is allocated. In particular, given any quota on pollution, there exists a quota equilibrium. 
Moreover, we establish a welfare theorem for our equilibrium concept: every quota equilibrium consumption-production pair is \emph{constrained Pareto optimal}, that is, it is Pareto optimal among all quota-compliant consumption-production pairs with the same total net pollution emission. As \cref{exampleproperty} indicates,
full Pareto optimality can often be achieved for quota equilibrium if the government sets the right quota. However, the allocation of emissions quota between the government firm and the private firms has an important impact on the distribution of consumption and welfare among the agents.  

We start this section by introducing the following characterization of agents' preferences as presented in \citet{hilden74}:
\begin{definition}\label{defpref}
The set $\mathcal{P}$ of preferences on an Euclidean space $\Reals^{\ell}$ consists of elements of the form $(X, \succ)$, where
\begin{itemize}
    \item The \emph{consumption set} $X\subset \NNReals^{\ell}$ is closed and convex;
    \item $\succ$ is a continuous and irreflexive preference relation defined on $X$. 
\end{itemize}
\end{definition}

For every $(X, \succ)\in \mathcal{P}$ and $a, b\in X$, $a\succ b$ means $a$ is strictly preferred to $b$. Note that we require neither completeness nor transitivity of $\succ$ in \cref{defpref}. 
A preference $\succ$ on $X$ is \emph{continuous} if, for every $x, y\in X$ with $x\succ y$, there exist relatively open sets $U\ni x$ and $V\ni y$ such that $a\succ b$ for all $a\in U$ and $b\in V$. The 
set $\mathcal{P}$ is equipped with the topology of closed convergence, which makes it a compact metric space as indicated in \citet{hilden74}.
For two elements $x_1, x_2\in \Reals^{\ell}$, we abuse notation and write $(x_1, x_2)\in (X, \succ)$ if $x_1,x_2\in X$ and $x_1\succ x_2$. A preference $P=(X, \succ)$ is \emph{convex} if $\{y\in X: y\succ x\}$ is convex for every $x\in X$, and we use $\mathcal{P}_{H}$ to denote the set of  convex preference from $\mathcal{P}$.
Let $\Delta=\{p\in \Reals^{\ell}: \|p\|=\sum_{k=1}^{l}|p_k|=1\}$ be the set of all prices.
\footnote{The set $\Delta$ is not convex. Thus, to establish the existence of equilibrium, the Kakutani fixed point theorem is not applicable. \citet{Hart75} established the existence of equilibrium with the price set $\Delta$ through a fixed and antipodal point theorem (see Section 2 of \citet{Hart75}).} 
Note that we allow for negative prices which can be interpreted as fees for emissions of bads.   
The definition of a finite production economy with quota is:

\begin{definition}\label{defmeasureeco}
A finite production economy with quota
\[
\mathcal E\equiv \{(X, \succ_{\omega}, P_{\omega}, e_{\omega}, \theta)_{\omega\in \Omega}, (Y_j)_{j\in J}, (m^{(j)})_{j\in J}, \mathcal{Z}(m)\} \nonumber
\] 
is a list  such that 
\begin{enumerate}[(i)]
\item \label{item-consumption_set} $\Omega$ is a finite set of agents. For every agent $\omega\in \Omega$, its consumption set $X(\omega)$ is a non-empty, closed and convex subset of $\NNReals^{\ell}$. We write $X_{\omega}$ for $X(\omega)$;

\item $J$ is a finite set of producers. Our economy has two types of firms: private firms and a single government firm. The government firm, denoted as firm $0$, has the production set $Y_{0}=\{0\}$. For each private firm $j\in J$, its production set $Y_{j}\subset \Reals^{\ell}$ is a non-empty subset. We write $Y=\prod_{j\in J}Y_{j}$; 

\item The set of allocations is 
$\mathcal{A}=\prod_{\omega\in \Omega}X_{\omega}$ is equipped with the product topology; 

\item\label{chaprefer} Let $M_{\omega}=\mathcal{A}\times Y\times \Delta\times X_{\omega}$ for every $\omega\in \Omega$. The \emph{global preference relation} of agent $\omega$ is $\succ_{\omega}\subset M_{\omega}\times M_{\omega}$. For $m,m'\in M_{\omega}$, we write $m\succ_{\omega} m'$ to mean that the agent $\omega$ strictly prefers $m$ over $m'$. 
The \emph{preference map} of agent $\omega$ is a map 
$P_{\omega}:\mathcal{A}\times Y\times \Delta\to \PowerSet(X_{\omega})\times \PowerSet(X_{\omega}\times X_{\omega})$\footnote{As the government firm's production set is the singleton $\{0\}$, the government firm's production has no impact on agents' preferences.} given by 
\[
P_{\omega}(x,y,p)=(X_{\omega},\{(a, b)\in X_{\omega}\times X_{\omega} | (x,y,p,a)\succ_{\omega}(x,y,p,b)\}). \nonumber
\]
For every $\omega\in \Omega$, $P_{\omega}$ satisfies:
\begin{itemize}
    \item The range of $P_{\omega}$ is a subset of $\mathcal{P}$. By \cref{defpref}, $P_{\omega}(x,y,p)$ can be written as $(X_{\omega}, \succ_{x,y,\omega,p})$;
    \item For $x, x'\in \mathcal{A}$ with $x(i)=x'(i)$ for all $i\neq \omega$, $P_{\omega}(x, y, p)=P_{\omega}(x', y, p)$ for all $(y, p)\in Y\times \Delta$;
    \item $P_{\omega}$ is continuous in the norm topology on $\mathcal{A}\times Y\times \Delta$;
    %\footnote{\label{preferftnote}Let $d$ be the metric induced by the norm topology on $\mathcal{L}^{1}(T\setminus \{t\}, \NNReals^{\ell})\times Y\times \Delta$ and $d_{\mathrm{cv}}$ denote the metric on $\mathcal{P}$ that generates the closed convergence topology. For all $\epsilon>0$ and all $(f, y, p)\in \mathcal{L}^{1}(T\setminus \{t\}, \NNReals^{\ell})\times Y\times \Delta$, there exists $\delta>0$ such that for all $(f', y', p')\in \mathcal{L}^{1}(T\setminus \{t\}, \NNReals^{\ell})\times Y\times \Delta$ with $d\big((f,y,p), (f',y',p')\big)<\delta$, we have $d_{\mathrm{cv}}\big(P_{t}(x,y,p), P_{t}(x',y',p')\big)<\epsilon$, where $x,x'$ are points in $\mathcal{L}^{1}(T, \NNReals^{\ell})$ such that $x(i)=f(i)$ and $x'(i)=f'(i)$ for all $i\neq t$. Note that, if $(T,\cB,\mu)$ is atomless, then this condition is equivalent to $P_{t}$ being continuous in the norm topology $\mathcal{L}^{1}(T, \NNReals^{\ell})\times Y\times \Delta$.};
\end{itemize}

%\item For each $t\in T$, $P_{t} : \mathcal{L}^{1}(T, \NNReals^{l})\times Y\times \Delta \rightarrow \mathcal{P}$ denotes the preference relation of agent $t$ where $P_{t}(x,y,p) = (X_t,\succ_{x,y,t,p})$. $P$ is continuous in the norm topology on $\mathcal{L}^{1}(T, \NNReals^{l})\times Y\times \Delta$ and measurable in $T$ \footnote{Let $f, f'\in \mathcal{L}^{1}(T, \NNReals^{l})$ be such that $f(i)=f'(i)$ for all $i\neq t$. As $(T,\cB,\mu)$ is an atomless probability space and $P$ is continuous in the norm topology on $\mathcal{L}^{1}(T, \NNReals^{l})\times Y\times \Delta$, we have $P_{t}(f, y, p)=P_{t}(f', y, p)$ for all $(y, p)\in Y\times \Delta$. }; 

\item $\theta(\omega)(j)$ is agent $\omega$'s share of firm $j$ such that $\sum_{\omega\in \Omega}\theta(\omega)(j)=1$ for all $\omega\in \Omega$. We sometimes write $\theta_{\omega j}$ for $\theta(\omega)(j)$. The agents' shares of private firms are exogenous to the model. The government chooses the agents' shares $\theta(\omega)(0)$ of the government firm; this choice allows the government to determine how any government revenues are rebated to agents;

\item $e\in (\NNReals^{\ell})^{\Omega}$ is the vector of initial endowments for the agents;

\item The government chooses to regulate the first $k\leq \ell$ commodities and assigns quotas on regulated commodities to the firms. For each $j\in J$, define $m^{(j)}\in \NegReals^{k}$ to be the negative of the quota for the firm $j$. Let $m=\sum_{j\in J}m^{(j)}$. 
The \emph{quota-compliance region}
$\mathcal{Z}(m)=\{m\}\times \prod_{k<n\leq \ell}\mathcal{Z}(m)_{n}$ is a convex subset of $\NegReals^{\ell}$, where $\mathcal{Z}(m)_{n}$ is either $\{0\}$ or $\NegReals$ for $k<n\leq \ell$. 
\end{enumerate}
\end{definition}

\begin{remark}\label{spremark}
The two main features of our model are:
\begin{enumerate}
   \item \cref{chaprefer} characterizes each agent's preference through the global preference relation $\succ_{\omega}$ and the preference map $P_{\omega}$. The preference relation $\succ_{\omega}$ represents the agent's preference on all agents' consumption, production, prices and her own consumption. The agent, however, has no control over other agent's consumption, production and prices. Hence, given all other agent's choices, production and prices, the agent chooses her bundle according to the preference map $P_{\omega}$. For the existence of equilibrium, one only needs to work with the preference map $P_{\omega}$. However, the preference relation $\succ_{\omega}$ is essential for studying welfare properties and potential Pareto improvement of quota-compliant consumption-production pairs.
   \item Our model incorporates a quota regulatory scheme by defining the quota-compliance region $\mathcal{Z}(m)$, which reflects the society's choice on which commodities to dispose, and at what quantity. $-m=-\sum_{j\in J}m^{(j)}$ is the quota on total net emissions of regulated commodities. $-m^{(j)}$ represents firm $j$'s property right to emit regulated commodities. The profits attributable to this property right flow through to the shareholders of firm $j$. Because the shareholdings of the government firm is chosen by the government, the government determines the distribution of those profits; it could, for example, choose to distribute those profits equally to all agents. The allocation of quotas between the government firm and the private firms thus plays a key role in the effect of the quota scheme on income distribution. In equilibrium, the emissions of a firm may be above or below that firm's quota; the emissions of the firms are determined endogenously. When the government firm's quota is $0$, our model captures the \emph{cap and trade} of pollution rights, in which the profits from reducing emissions accrue to the shareholders of the firms that reduce those emissions. When private firms' quotas are $0$, we refer to our model as \emph{global quota economy}. 
   In this case, the profits from reducing emissions accrue to agents, according to the shares the government firm assigned to them by the government. The allocation of quotas between the government firm and the private firms may be different for different regulated commodities; for example, the government might impose cap and trade on one regulated commodity, while assigning the quota for another regulated commodity entirely to the government firm. The government may impose a hybrid model, in which both the government firm and the private firms have non-zero quotas. 
   \item As $\mathcal{Z}(m)_{n}$ is either $\NegReals$ or $\{0\}$ for $k<n\leq \ell$, our model allows for the government to allow free disposal, or prohibit disposal, for non-regulated commodities. 
   The government may choose to tolerate some bads, and some goods (such as atmospheric oxygen) to be left unused at price zero.\footnote{In the existing literature, oxygen would be said to be ``disposed."} On the other hand, the government may choose to prohibit all emissions of certain commodities.\footnote{CFCs and HCFCs are bads which deplete ozone, and are now prohibited under the Montreal protocol.} In addition, some goods are fully consumed at a non-negative price. 
\end{enumerate}
\end{remark}

For every $\omega\in \Omega$, $p\in \Delta$ and $y\in Y$, the \emph{quota budget set} $B_{\omega}^{m}(y, p)$ is defined as
\[
\{z\in X_{\omega}: p\cdot z \leq p\cdot e(\omega) + \sum_{j\in J}\theta_{\omega j}\big(p\cdot y(j)+\pi_{k}(p)\cdot m^{(j)}\big)\}.\nonumber
\]
For each private firm $j\in J$, since the firm can emit the first $k$ commodities freely up to its quota $m^{(j)}$, the firm's profit at a given price $p$ is $p\cdot y(j)+\pi_{k}(p)\cdot m^{(j)}$.\footnote{If a private firm emits less than its quota, then the firm generates additional revenue by selling its remaining quota to other firms. If a private firm emits more than its quota, then the firm needs to purchase quota from other firms.} The government firm's profit comes solely from selling its quota. In particular, the government firm's profit at a given price $p$ is $p\cdot y(0)+\pi_{k}(p)\cdot m^{(0)}=\pi_{k}(p)\cdot m^{(0)}$ since the government firm's production set is the singleton $\{0\}$. Hence, the agent's budget consists of the value of her endowment and her dividend from firms. 
For each $\omega\in \Omega$ and $(x,y,p)\in \mathcal{A}\times Y\times \Delta$, the \emph{quota demand set} $D^{m}_{\omega}(x,y,p)$ consists of all elements in $B_{\omega}^{m}(y, p)$ that maximize the agent's preference given $(x,y,p)$. 
In particular, $D^{m}_{\omega}(x,y,p)$ is defined as:
\[
\{z \in B^{m}_{\omega}(y,p): w \succ_{x,y,\omega,p} z\implies w\not\in B^{m}_{\omega}(y,p)\}.\nonumber
\]
Given a price $p$, for each $j\in J$, the firm's supply set $S_{j}^{m}(p)$ is $\argmax_{z\in Y_{j}}\big(p\cdot z+\pi_{k}(p)\cdot m^{(j)}\big)$. As $\pi_{k}(p)\cdot m^{(j)}$ does not depend on the firm's production plan, we have $S_{j}^{m}(p)=\argmax_{z\in Y_{j}}p\cdot z$. 
All firms' profits depend only on prices and their own production.\footnote{We assume producers are profit maximizers. \citet{makarov81} established a general equilibrium existence theorem which allows for firm objectives other than profit maximization.} 

Free-disposal equilibrium and non-free-disposal equilibrium are two classical general equilibrium notions, and they differ only in the resource feasibility constraint. While free-disposal equilibrium only requires demand to be no more than the supply, non-free-disposal equilibrium requires demand to be exactly equal to the supply for each commodity. Free-disposal equilibrium allows for the bad to be freely disposed of and hence precludes any control on total net pollution emission.
While non-free-disposal equilibrium is widely used in GE models with bads,\footnote{\citet{Mck59} used the non-free-disposal equilibrium with possible negative prices, which is followed by \citet{Berg71}, \citet{Hart75}, \citet{pole93}, and others.} it does not allow one to limit net pollution emission to a positive amount. While it might be desirable to eliminate $\mathrm{CO}_{2}$ emissions at some future date, it is not practical to do so in the near term. 
Our equilibrium notion generalizes these two classical equilibrium notions by allowing for positive quotas on certain pollution emissions.

\begin{definition}\label{def_mameqprod}
Let $\mathcal{E}=\{(X, \succ_{\omega}, P_\omega, e_\omega, \theta)_{\omega\in \Omega}, (Y_j)_{j\in J}, (m^{(j)})_{j\in J}, \mathcal{Z}(m)\}$ be a finite production economy with quota. 
A $\mathcal{Z}(m)$-compliant quota equilibrium is $(\bar{x}, \bar{y}, \bar{p})\in \mathcal{A}\times Y\times \Delta$ such that the following conditions are satisfied:
\begin{enumerate}
\item $\bar{x}(\omega)\in D^{m}_{\omega}(\bar{x},\bar{y},\bar{p})$ for all $\omega\in \Omega$;
\item $\bar{y}(j)\in S^{m}_{j}(\bar{p})$ for all $j\in J$. Every firm is profit maximizing given the price $\bar{p}$;
\item $\sum_{\omega\in \Omega}\bar{x}(\omega)-\sum_{\omega\in \Omega}e(\omega)-\sum_{j\in J}\bar{y}(j)\in \mathcal{Z}(m)$.
\end{enumerate} 
\end{definition}

As we have emphasized in the introduction, the equilibrium price $\bar{p}$ emerges endogenously\footnote{In particular, the equilibrium price of the quota is determined through market forces.} after the government has set the quota (along with its allocation) and shareholdings of the government firm.
Our feasibility constraint $\sum_{\omega\in \Omega}\bar{x}(\omega)-\sum_{\omega\in \Omega}e(\omega)-\sum_{j\in J} \bar{y}(j)\in \mathcal{Z}(m)$ implies, at equilibrium, that total net emission of the regulated commodities equals the pre-specified total quota, which is the aggregation of the government firm's quota and private firms' quota. 
The existence and the welfare properties of quota equilibrium are studied in the rest of this section.  

\subsection{Existence of Quota Equilibrium}\label{secquotaeq}
In this section, for any given quota, we establish the existence of a quota equilibrium for production economies with quota as in \cref{defmeasureeco}. 
Our result generalizes Proposition 3.2.3 in \citet{fl03bk}, in which the quota-compliance region is a convex cone. As we have discussed in the introduction, while quota-compliance cone allows the society to prohibit the emissions of a given pollutant, it does not allow for setting a positive cap on net pollution emissions. However, Proposition 3.2.3 in \citet{fl03bk} plays a crucial role in establishing our main existence theorem, and a version of Proposition 3.2.3 tailored for our setting is proved in the \cref{appendixcone} (see \cref{finitemainexpd}).
The proof of the existence of quota equilibrium consists of the following major steps:
\begin{enumerate}
    \item We construct a production economy with quota $\mathcal{E}'$ by shifting each firm's production set by its prespecified quota. The quota-compliance region of $\mathcal{E}'$ is a cone;
    \item We apply \cref{finitemainexpd} to show that $\mathcal{E}'$ has an equilibrium;
    \item Finally, we show every equilibrium of $\mathcal{E}'$ is a quota equilibrium of $\mathcal{E}$. 
\end{enumerate}
To state and prove our main existence theorem, we start with the definition of quota-compliant consumption-production pairs, quota-compliant consumption set and quota-compliant production set for individual firms.

\begin{definition}\label{def_attainable}
The set of quota-compliant consumption-production pair of $\mathcal{E}$ is
\[
\cO=\left\{(x,y)\in \mathcal{A}\times Y: \sum_{\omega\in \Omega}x(\omega)-\sum_{\omega\in \Omega}e(\omega)-\sum_{j\in J}y(j)\in \mathcal{Z}(m)\right\}. \nonumber
\]
The set $\hat Y_j$ of quota-compliant production plans for the $j$-th producer is 
\[
\left\{y_j\in Y_j : \exists (x, y')\in \cA\times \prod_{i\neq j}Y_{i}, \sum_{\omega\in \Omega} x(\omega)-\sum_{\omega\in \Omega} e(\omega)- y_j-\sum_{i\neq j}{y'}(i)\in \mathcal{Z}(m)\right\}. \nonumber
\] 
The set $\hat{X}_{i}$ of quota-compliant consumption for the $i$-th agent is
\[
\left\{x_i\in X_{i} : \exists (x', y)\in \prod_{\omega\neq i}X_{\omega}\times \prod_{j\in J}Y_{j}, x_{i}+\sum_{\omega\neq i} x'(\omega)-\sum_{\omega\in \Omega} e(\omega) -\sum_{j\in J}y(j)\in \mathcal{Z}(m)\right\}. \nonumber
\]
\end{definition}
We now state and prove our main existence result.

\begin{customthm}{1}\label{standardqtmain}
Let $\mathcal E=\{(X, \succ_{\omega}, P_\omega, e_\omega, \theta)_{\omega\in \Omega}, (Y_j)_{j\in J}, (m^{(j)})_{j\in J}, \mathcal{Z}(m)\}$ be a finite production economy with quota as in \cref{defmeasureeco}. Suppose $\mathcal{E}$ satisfies the following conditions: 
\begin{enumerate}[(i)]
    \item for $\omega\in \Omega$,  $\hat{X}_{\omega}$ is compact and the preference map $P_{\omega}$ takes value in $\mathcal{P}_{H}$;\footnote{As noted in \citet{fl03bk}, this condition can be weakened to the following condition: for each $(x,y,p)\in \cO\times (\Delta\cap \mathcal{Z}^{0})$ and all $\omega\in \Omega$, $(x(\omega),x(\omega))\not\in \conv(P_{\omega}(x,y,p))$, where $\conv(P_{\omega}(x, y,p))$ denotes the convex hull of $P_{\omega}(x,y,p)$.}
    \item\label{esto0} there exists an agent $\omega_0\in \Omega$ such that the set $X_{\omega_0} - \sum_{j\in J}\theta_{\omega_{0} j}(Y_j+\{E(m^{(j)})\})$\footnote{For all $j\in J$, $E(m^{(j)})\in \NegReals^{\ell}$ is the vector such that its projection to the first $k$th coordinates is $m^{(j)}$ and its other coordinates are $0$.} has non-empty interior $U_{\omega_0}\subset \Reals^{\ell}$ and $e(\omega_0)\in U_{\omega_0}$;
\item \label{conditionexplain}there exists 
a commodity $s\in \{1,2,\dotsc,\ell\}$\footnote{The commodity $s$ will typically be a non-regulated commodity.} such that:
\begin{itemize}
    \item the projection $\pi_{s}(X_{\omega_0})$ is unbounded, and the agent $\omega_0$ has a strongly monotone preference on the commodity $s$,\footnote{Given any $(x,y,p)\in \mathcal{A}\times Y\times \Delta$, we have $(a,a')\in P_{\omega_0}(x,y,p)$ if $a_s>a'_s$ and $a_t=a'_t$ for all $t\neq s$.} where $\omega_0$ is the agent specified in \cref{esto0};
    \item for every $\omega\in \Omega$, there is an open $V_{\omega}\subset \Reals$ containing the $s$-th coordinate $e(\omega)_{s}$ of $e(\omega)$ such that $(e(\omega)_{-s},v)\in X_{\omega} - \sum_{j\in J}\theta_{\omega j}(Y_j+\{E(m^{(j)})\})$\footnote{$(e(\omega)_{-s},v)$ is the vector such that its $s$-th coordinate is $v$, and its $t$-th coordinate is the same as the the $t$-th coordinate of $e(\omega)$ for all $t\neq s$.} for all $v\in V_{\omega}$;\footnote{\cref{esto0} and \cref{conditionexplain} of \cref{standardqtmain} are needed to establish the existence of a quota equilibrium. However, we only need $e(\omega)\in X_{\omega}-\sum_{j\in J}\theta_{\omega j}(Y_j+\{E(m^{(j)})\})$ for every $\omega\in \Omega$, not  \cref{esto0} and \cref{conditionexplain}, to establish the existence of a quota quasi-equilibrium. See Proposition 3.2.3 of \citet{fl03bk} for a proof of the special case where the quota is $0$.}
\end{itemize}
    \item \label{satiate}for $\omega\in \Omega$ and $(x,y)\in \cO$, there is $u\in X_\omega$ with $(u, x_\omega)\in \bigcap_{p\in \Delta\cap (\mathcal{Z}')^{0}}P_{\omega}(x,y,p)$, where 
    $\mathcal{Z}'=\{v\in \Reals^{\ell}: (\forall n\leq k)(v_n=0)\wedge (\forall k<n\leq \ell)(v_{n}\in \mathcal{Z}(m)_{n})\}$
    and $(\mathcal{Z}')^{0}=\{p\in \Delta: (\forall z\in \mathcal{Z}')(p\cdot z\leq 0)\}$ is the polar cone of $\mathcal{Z}'$;
    \item the aggregate production set $\bar{Y}=\left\{\sum_{j\in J}{y}(j): y\in Y\right\}$ is closed and convex, and for each $j\in J$, the set $\hat{Y}_{j}$ is relatively compact;
    \item the set $\bar{Y}+\mathcal{Z}(m)$ is closed. 
\end{enumerate}
Then, $\cO$ is non-empty, i.e., it is feasible to achieve the quota, and there exists a quota equilibrium. Moreover, the equilibrium price is in $(\mathcal{Z}')^{0}$.
\end{customthm}
\begin{remark}
\cref{standardqtmain} is a generalization of Proposition 3.2.3 in \citet{fl03bk} since we do not require the quota-compliance region $\mathcal{Z}(m)$ to be a cone. The conditions of \cref{standardqtmain} are similar to classical assumptions in the existing GE literature: 
\begin{enumerate}
\item Closedness and convexity are classical assumptions on consumption sets and on production sets. These conditions, in particular, imply that commodities are perfectly divisible. Compactness of quota-compliant consumption sets and relative compactness of quota-compliant production sets play essential roles in establishing the existence of equilibrium, as in \citet{de59}. Proposition 2.2.4 in \citet{fl03bk} provides sufficient conditions to guarantee the compactness of quota-compliant consumption sets and relative compactness of quota-compliant production sets;\footnote{For $A\subset \Reals^{\ell}$, let $r(A)$ be the recession cone of $A$. Suppose we have $r(\sum_{\omega\in \Omega}X_{\omega})\cap r(\bar{Y}+\mathcal{Z}(m))=\{0\}$, $r(\sum_{\omega\in \Omega}X_{\omega})\cap \big(-r(\sum_{\omega\in \Omega}X_{\omega})\big)=\{0\}$ and $r(\bar{Y})\cap \big(-r(\bar{Y})\big)=\{0\}$, by Proposition 2.2.4 in \citet{fl03bk}, quota-compliant consumption sets and production sets are compact and relatively compact, respectively.}

\item \citet{fl03bk} does not require firms' production sets to contain zero. As a result, a firm may incur a loss at equilibrium, and this loss is passed on to the shareholders. In that sense, \citet{fl03bk}'s firms are not limited liability. 
In this paper, we follow \citet{fl03bk}'s formulation. In our setting, the relevant production set for firm $j$ is $Y_{j}+E(m^{(j)})$, and we cannot rule out the possibility that some regulated commodities have positive equilibrium price. Thus, even if $0$ is in $Y_j$, we cannot guarantee that firm $j$'s profit is non-negative.   
It is essential for the existence of equilibrium that firm profits, whether positive or negative, be passed on to the agents in some fashion. If one wished to make the firms limited liability, one would need to introduce a bankruptcy mechanism that specifies who bears firm losses. 
In all our examples, the equilibrium prices for regulated commodities are negative and $0$ is an element of all production sets, which jointly imply that firms' profit at equilibrium are non-negative, so the bankruptcy issue does not arise;

\item Each agent's preference map $P_{\omega}$ is assumed to be continuous in the closed convergence topology as in \cref{defmeasureeco}. As a result, every agent's preference correspondence (which defines the strictly preferred set) is lower hemicontinuous. If the preferences are price-independent, \cref{satiate} in the assumptions of \cref{standardqtmain} is equivalent to assuming non-satiation for the set of quota-compliant consumption-production pairs;

\item The survival assumption $e_{\omega}\in X_\omega-\sum_{j\in J}\theta_{\omega j}(Y_j+\{E(m^{(j)})\})$ in our setting is derived from the survival assumption that is commonly used in the GE literature to ensure the existence of a quasi-equilibrium.\footnote{The survival assumption implies that an agent can survive without participating in any exchanges using her initial endowment and shares in production. In particular, an agent who supplies labor in an equilibrium is able to survive, and hence supply that labor.} 
Following the previous literature, we could have strengthened our survival assumption to $e_{\omega}\in \mathrm{int}\big(X_\omega-\sum_{j\in J}\theta_{\omega j}(Y_j+\{E(m^{(j)})\})\big)$\footnote{As in the previous literature, the interior is taken with respect to the topology of $\Reals^{\ell}$, not with respect to the subspace topology. Hence, the strengthened survival assumption implies the set $X_\omega-\sum_{j\in J}\theta_{\omega j}(Y_j+\{E(m^{(j)})\})$ has non-empty interior in $\Reals^{\ell}$.} to guarantee that every quota quasi-equilibrium is a quota equilibrium.\footnote{There exist, however, a few papers relaxing the assumption $e_{\omega}\in \mathrm{int}\big(X_\omega-\sum_{j\in J}\theta_{\omega j}(Y_j+\{E(m^{(j)})\})\big)$. See e.g. \citet{McKenzie81}, \citet{de62}, \citet{arrow71}, \citet{Berg71} and \citet{Florig99}. Also see \citet{fl03bk} for a detailed discussion.} 
This strengthened survival assumption is not satisfied if there is an agent with no shareholding of any private firm and the projection of the agent's consumption set to some coordinate is a singleton, and hence is economically restrictive.\footnote{Poor people generally do not have any shareholdings of private firms. Moreover, as poor individuals are incapable of sequestering emissions of regulated commodities, it is reasonable to assume that the projections of their consumption sets to regulated commodities are singletons, see \cref{noconsumefoot} for further explanation.}

\cref{esto0} and \cref{conditionexplain} are less restrictive, as they allow for the existence of agents who are not endowed with certain commodities, have no shareholdings of private firms and are unable to consume regulated commodities.
\cref{esto0} requires there be a single agent $\omega_0$ whose endowment is in the interior of $X_{\omega_0}-\sum_{j\in J}\theta_{\omega_0 j}(Y_j+\{E(m^{(j)})\})$
\footnote{\cref{esto0} is generally satisfied if there is a rich individual $\omega_0$ who has positive shareholdings of a set of private firms that generate all types of regulated commodities.}, which implies that the agent $\omega_0$'s budget at any quota quasi-equilibrium is strictly positive. \cref{esto0} and the first bullet of \cref{conditionexplain} imply that the quota quasi-equilibrium price of the commodity $s$ is strictly positive. Hence, by the second bullet of \cref{conditionexplain}, every agent has a positive budget at any quota quasi-equilibrium, which leads to the conclusion that every quota quasi-equilibrium is a quota equilibrium;\footnote{We present a complete argument in \cref{secquasieq}.}

\item The assumption that $\bar{Y}+\mathcal{Z}(m)$ is closed follows from $\bar{Y}\cap (-\mathcal{Z}(m))=\{0\}$, as indicated in Proposition 2.2.4 in \citet{fl03bk}.
Note that $(-\mathcal{Z}(m))\subset \NNReals^{\ell}$. If there is no free production and each firm can choose not to produce at all, then we have $\bar{Y}\cap (-\mathcal{Z}(m))=\{0\}$.
\end{enumerate}
\end{remark}
\begin{proof}[Proof of ~\cref{standardqtmain}]
By \cref{esto0} and the second bullet of \cref{conditionexplain} in the assumptions of \cref{standardqtmain}, we have $e(\omega)\in X_{\omega}-\sum_{j\in J}\theta_{\omega j}(Y_j+\{E(m^{(j)})\})$ for all $\omega\in \Omega$. So the set $\cO$ of quota-compliant consumption-production pairs is non-empty, hence is feasible to achieve the quota. 
Let $\mathcal{E}'=\{(X, \succ'_{\omega}, P'_{\omega}, e_{\omega}, \theta)_{\omega\in \Omega}, (Y'_j)_{j\in J}, (\tilde{m}^{(j)}), \mathcal{Z'}\}$ be a finite production economy with quota: 
\begin{enumerate}
    \item $Y'_{j}=Y_{j}+\{E(m^{(j)})\}$ for all $j\in J$. Let $Y'=\prod_{j\in J'}Y'_{j}$; 
    \item Each firm's quota $\tilde{m}^{(j)}$ is $0$;
    \item Let $\mathcal{Z}'=\{v\in \Reals^{\ell}: (\forall n\leq k)(v_n=0)\wedge (\forall k<n\leq \ell)(v_{n}\in \mathcal{Z}(m)_{n})\}$;
    \item We only provide a rigorous definition of the induced preference map $P'_{\omega}$ \footnote{The agent's global preference relation $\succ'_{\omega}$ is defined similarly. To establish the existence of a quota equilibrium, one only needs to work with the preference map $P'_{\omega}$.}. For $y\in Y'$, let $y(\mathcal{E})\in Y$ be such that $y(\mathcal{E})_{j}=y_{j}-E(m^{(j)})$ for all $j\in J$. 
    For $\omega\in \Omega$, the preference map $P'_{\omega}: \mathcal{A}\times Y'\times \Delta\to \mathcal{P}$ is given by
    \[
    P'_{\omega}(x,y,p)=(X_{\omega},\{(a, b)\in X_{\omega}\times X_{\omega} |(x,y(\mathcal{E}),p,a)\succ_{\omega}(x,y(\mathcal{E}),p,b)\}).\nonumber
    \]  
\end{enumerate}
To show that the derived economy $\mathcal{E}'$ has a $\mathcal{Z}'$-compliant quota equilibrium, we must verify that $\mathcal{E}'$ satisfies the assumptions of \cref{finitemainexpd}. It is easy to see that: 
\begin{enumerate}
    \item By the construction of $P'_{\omega}$, $P'_{\omega}$ takes value in $\mathcal{P}_{H}$ for all $\omega\in \Omega$. By \cref{esto0} and \cref{conditionexplain}, we have $e_{\omega}\in X_\omega-\sum_{j\in J}\theta_{\omega j}(Y_j+\{E(m^{(j)})\})$. Hence, we have $e(\omega)\in X_{\omega}-\sum_{j\in J}\theta_{\omega j}Y'_{j}$ for all $\omega\in \Omega$;
    \item there exists an agent $\omega_0\in \Omega$ such that the set $X_{\omega_0} - \sum_{j\in J}\theta_{\omega_{0} j}Y'_j$ has non-empty interior $U_{\omega_0}\subset \Reals^{\ell}$ and $e(\omega_0)\in U_{\omega_0}$;
\item there exists 
a commodity $s\in \{1,2,\dotsc,\ell\}$ such that:
\begin{itemize}
    \item the projection $\pi_{s}(X_{\omega_0})$ is unbounded, and the agent $\omega_0$ has a strongly monotone preference on the commodity $s$;
    \item for every $\omega\in \Omega$, there is an open set $V_{\omega}$ containing the $s$-th coordinate $e(\omega)_{s}$ of $e(\omega)$ such that $(e(\omega)_{-s},v)\in X_{\omega} - \sum_{j\in J}\theta_{\omega j}Y'_j$ for all $v\in V_{\omega}$;
\end{itemize}  
    \item The projection of $\mathcal{Z}'$ to the first $k$-th coordinates is $\{0\}$.
\end{enumerate}
For every $i\in \Omega$, let $\hat{X}'_{i}$ be the set of quota-compliant consumption for the $i$-th agent for the economy $\mathcal{E}'$. In particular, $\hat{X}'_{i}$ takes the form of: 
\[
\left\{x_{i}\in X_{i} : \exists (\tilde{x}, y)\in \prod_{\omega\neq i}X_{\omega}\times \prod_{j\in J}Y'_{j}, x_{i}+\sum_{\omega\neq i} \tilde{x}(\omega)-\sum_{\omega\in \Omega} e(\omega) -\sum_{j\in J}y(j)\in \mathcal{Z}'\right\}. \nonumber
\]
It is straightforward to verify that $\hat{X}'_{i}$ is the same set as $\hat{X}_{i}$. Hence, $\hat{X}'_{\omega}$ is compact for all $\omega\in \Omega$ since $\hat{X}_{\omega}$ is compact for all $\omega\in \Omega$.

\begin{claim}\label{eattainmod}
For  all $\omega\in \Omega$ and all quota-compliant consumption-production pairs $(x, y)$ of $\mathcal{E}'$, there exists $u\in X_{\omega}$ so that $(u, x(\omega))\in \bigcap_{p\in \Delta\cap (\mathcal{Z}')^{0}}P'_{\omega}(x,y,p)$. 
\end{claim}
\begin{proof}
Fix $\omega\in \Omega$.
Let $(x, y)$ be a quota-compliant consumption-production pair of $\mathcal{E}'$. Then $(x, y(\mathcal{E}))$ is a quota-compliant consumption-production pair of $\mathcal{E}$. So there exists $u\in X_{\omega}$ such that $(u, x(\omega))\in \bigcap_{p\in \Delta\cap (\mathcal{Z}')^{0}}P_{\omega}(x,y(\mathcal{E}),p)$. As $P'_{\omega}(x,y,p)=P_{\omega}(x,y(\mathcal{E}),p)$ for all $p\in \Delta$, we have $(u, x(\omega))\in \bigcap_{p\in \Delta\cap (\mathcal{Z}')^{0}}P'_{\omega}(x,y,p)$.
\end{proof}

\begin{claim}\label{barYclaim}
The aggregate production set $\bar{Y}'=\sum_{j\in J}Y'_{j}$ is closed and convex. For each $j\in J$, the set $\hat{Y}'_{j}$ of quota-compliant production plans for the $j$-th producer for the economy $\mathcal{E}'$ is relatively compact. Finally, the set $\bar{Y}'+\mathcal{Z}'$ is closed. 
\end{claim}
\begin{proof}
Note that $\bar{Y}'=\bar{Y}+\{E(m)\}$, where $E(m)\in \NegReals^{\ell}$ is the vector such that its projection to the first $k$th coordinates is $m$ and its rest coordinates are $0$. As $\{E(m)\}$ is a singleton and $\bar{Y}$ is closed and convex, $\bar{Y}'$ is closed and convex. Note that $\bar{Y}'+\mathcal{Z}'=\bar{Y}+\{E(m)\}+\mathcal{Z}'=\bar{Y}+\mathcal{Z}(m)$. Hence, we know that $\bar{Y}'+\mathcal{Z}'$ is closed. For every $j\in J$, the set $\hat{Y}'_{j}$ of quota-compliant production plans for the $j$-th producer for the economy $\mathcal{E}'$ is:
\[
\left\{y_j\in Y'_j : \exists (x, \tilde{y})\in \cA\times \prod_{i\neq j}Y'_{i}, \sum_{\omega\in \Omega} x(\omega)-\sum_{\omega\in \Omega} e(\omega)- y_j-\sum_{i\neq j}\tilde{y}(i)\in \mathcal{Z}'\right\}. \nonumber
\] 
Note that an element $y_j\in \hat{Y}'_{j}$ if and only if $y_j-E(m^{(j)})\in \hat{Y}_{j}$. As $\hat{Y}_{j}$ is relatively compact for every $j\in J$, $\hat{Y}'_{j}$ is relatively compact for every $j\in J$.  
\end{proof}
By \cref{finitemainexpd},
there is a $\mathcal{Z}'$-compliant quota equilibrium $(\bar{x}, \bar{y}, \bar{p})$ for $\mathcal{E}'$ where $\bar{p}\in (\mathcal{Z}')^{0}$.
We now show that $(\bar{x}, \bar{y}(\mathcal{E}), \bar{p})$ is a quota equilibrium for $\mathcal{E}$:
\begin{enumerate}
    \item Note that we have $\bar{p}\cdot \bar{y}(j)=\bar{p}\cdot \bar{y}(\mathcal{E})(j)+\pi_{k}(\bar{p})\cdot m^{(j)}$. For every $j\in J$, we have $\bar{y}(j)\in \argmax_{z\in Y'_{j}}\bar{p}\cdot z$. As $\pi_{k}(\bar{p})\cdot m^{(j)}$ is a constant over $Y_{j}$, we have $\bar{y}(\mathcal{E})(j)\in S^{m}_{j}(\bar{p})$ for all $j\in J$;
    \item As $\sum_{\omega\in \Omega}\bar{x}(\omega)-\sum_{\omega\in\Omega}e(\omega)-\sum_{j\in J} \bar{y}(j)\in \mathcal{Z}'$, we have 
    \[
    &\sum_{\omega\in \Omega}\bar{x}(\omega)-\sum_{\omega\in \Omega}e(\omega)-\sum_{j\in J} \bar{y}(\mathcal{E})(j) \nonumber\\ 
    &=\sum_{\omega\in \Omega}\bar{x}(\omega)-\sum_{\omega\in \Omega}e(\omega)-\sum_{j\in J} \bar{y}(j)+E(m)\in \mathcal{Z}(m). \nonumber
    \]
\end{enumerate}
\begin{claim}\label{derebatein}
$\bar{x}(\omega)\in D^{m}_{\omega}(\bar{x},\bar{y}(\mathcal{E}),\bar{p})$ for all $\omega\in \Omega$. 
\end{claim}
\begin{proof}
Note that $\bar{p}\cdot \bar{y}(j)=\bar{p}\cdot \bar{y}(\mathcal{E})(j)+\pi_{k}(\bar{p})\cdot m^{(j)}$ for all $j\in J$. 
Thus, for all $\omega\in \Omega$, the quota budget set $B'_{\omega}(\bar{y}, \bar{p})$ for agent $\omega$ of the economy $\mathcal{E}'$ can be written as: 
\[
\left\{z\in X_{\omega}: \bar{p}\cdot z \leq \bar{p}\cdot e(\omega)+\sum_{j\in J}\theta_{\omega j}\big(\bar{p}\cdot \bar{y}(\mathcal{E})(j)+\pi_{k}(\bar{p})\cdot m^{(j)}\big)\right\}, \nonumber
\]
which is the same as the quota budget set $B^{m}_{\omega}(\bar{y}(\mathcal{E}),\bar{p})$ of the economy $\mathcal{E}$. As $P_{\omega}(\bar{x},\bar{y}(\mathcal{E}),\bar{p})=P'_{\omega}(\bar{x},\bar{y},\bar{p})$ for all $\omega\in \Omega$, the quota demand set $D'_{\omega}(\bar{x},\bar{y}, \bar{p})$ for agent $\omega$ of the economy $\mathcal{E}'$ is the same as the quota demand set $D^{m}_{\omega}(\bar{x},\bar{y}(\mathcal{E}),\bar{p})$ of the economy $\mathcal{E}$. Thus, we conclude that $\bar{x}(\omega)\in D^{m}_{\omega}(\bar{x},\bar{y}(\mathcal{E}),\bar{p})$ for all $\omega\in \Omega$.
\end{proof}
By \cref{derebatein},  $(\bar{x}, \bar{y}(\mathcal{E}), \bar{p})$ is a quota equilibrium for $\mathcal{E}$ where $\bar{p}\in (\mathcal{Z}')^{0}$.
\end{proof}

\cref{standardqtmain} shows that the production economy with quota model, defined in \cref{defmeasureeco}, has a quota equilibrium under moderate assumptions. At any equilibrium, the total net emission of the first $k$ commodities equals the pre-specified quota. 

We now turn our attention to allocation of quotas between the government and private firms and how it affects agents' equilibrium consumption. 

\begin{customexp} {1}\label{exampleproperty}\normalfont
Let $\mathcal{E}=\{(X, \succ_{\omega}, P_\omega, e_\omega, \theta)_{\omega\in \Omega}, (Y_j)_{j\in J}, (m^{(j)})_{j\in J}, \mathcal{Z}(m)\}$ be a production economy with quota such that:
\begin{enumerate}
\item The economy $\mathcal{E}$ has three commodities $\mathrm{CO}_{2}$, coal and electricity, which we denote by $c_1$, $c_2$ and $c_3$;

\item There are two agents with consumption sets $X_{1}=X_{2}=\{0\}\times \NNReals^{2}$ and endowments $e(1)=e(2)=(0,1,0)$. 
Given the total net emission $v$ of $\mathrm{CO}_{2}$, the two agents have the same utility function $u_{v}(c_1,c_2,c_3)=c_3-\frac{v^{2}}{2}$;
    
\item The government firm's production set $Y_0$ is the singleton $\{0\}$. There is a single private firm with the production set $Y=\{(r,-r,r): r\in \NNReals\}$. So the private firm has the production technology to burn $r$ units of coal to generate $r$ units of electricity and $r$ units of $\mathrm{CO}_{2}$ as byproduct;

\item Let $m$ be the negative of the quota on total net emission of $\mathrm{CO}_{2}$. 
The quota-compliance region $\mathcal{Z}(m)$ is $\{m\}\times \NegReals\times \{0\}$; we require non-free-disposal of electricity while allowing for free disposal of coal at equilibrium. 
The quota $m^{(0)}$ for the government firm and the quota $m^{(1)}$ for the private firm will be specified later;

\item The government determines that the profit of the government firm is distributed equally to all agents. That is, we have $\theta(1)(0)=\theta(2)(0)=\frac{1}{2}$. The private firm is completely owned by the first agent. That is, we have $\theta(1)(1)=1$ and $\theta(2)(1)=0$. 
\end{enumerate}

Note that the two consumers are identical except for their shareholdings of the private firm, and the economy $\mathcal{E}$ can at most produce $2$ units of $\mathrm{CO}_{2}$. 
By \cref{standardqtmain}, $\mathcal{E}$ has a quota equilibrium for every $m\in [-2, 0]$ and every pair of $(m^{(0)},m^{(1)})\in [-2,0]\times [-2,0]$ with $m^{(0)}+m^{(1)}=m$. Since no electricity will be produced if we set the quota to be $0$, we focus on the case where $m\in (-2, 0)$.\footnote{Since no firm has $\mathrm{CO}_{2}$ sequestration technology, \cref{esto0} of \cref{standardqtmain} is satisfied if and only if the quota is not $0$. In fact, \cref{esto0} is usually satisfied if, for each regulated commodity, either the quota is positive or there is a firm that has the technology to sequester the commodity. 
In this example, the equilibrium price, the equilibrium production of the private firm and agents' equilibrium consumption on electricity are uniquely determined for $m\in (-2,0)$. This is not the case if $m=-2$ or $m=0$.}
We shall show that the allocation of quotas between the government firm and private firms has a significant impact on the welfare of consumers at the resulting equilibrium. In particular, when the entire quota is allocated to the government firm, the two agents' equilibrium consumptions are the same. On the other hand, when the entire quota is allocated to the private firm, at any quota equilibrium, the first agent consumes all of the electricity and the second agent consumes nothing.

We first consider the case where all quota for $\mathrm{CO}_{2}$ emission is allocated to the government firm, i.e., we consider a model of global quota on $\mathrm{CO}_{2}$ emission.  
\begin{claim}\label{governquota}
Let $m\in (-2,0)$. 
Suppose $m^{(0)}=m$ and $m^{(1)}=0$. Then, for any equilibrium $(\bar{x},\bar{y},\bar{p})$, the private firm's equilibrium production is $(-m,m,-m)$, the equilibrium price is $\bar{p}=(-\frac{1}{2},0,\frac{1}{2})$ and both agents' equilibrium consumptions of electricity are $\frac{-m}{2}$.
\end{claim}

The proof of \cref{governquota} is presented in \cref{appendixexp1}.
We now consider the case where all quota for $\mathrm{CO}_{2}$ emission is allocated to the private firm, i.e., our model captures cap and trade on $\mathrm{CO}_{2}$ emission rights.  
As we will see, this results in equilibrium with quite different consumptions by the two agents; the agent who owns the first firm captures all of the consumption of electricity. 

\begin{claim}\label{privatequota}
Let $m\in (-2,0)$. 
Suppose $m^{(0)}=0$ and $m^{(1)}=m$. Then, for any equilibrium $(\bar{x},\bar{y},\bar{p})$,
the private firm's production is $(-m,m,-m)$, the equilibrium price is $\bar{p}=(-\frac{1}{2},0,\frac{1}{2})$ and the equilibrium consumption of electricity of the first agent is $-m$ while the equilibrium consumption of electricity of the second agent is $0$. 
\end{claim}
The proof of \cref{privatequota} is presented in \cref{appendixexp1}. \cref{governquota} and \cref{privatequota} jointly imply that allocation of quota has a major impact on the agent's welfare.

\end{customexp}

\cref{exampleproperty} shows agents' equilibrium consumption can differ vastly under different quota allocations. On the other hand, every quota equilibrium can be realized as a global quota equilibrium by carefully setting the 
shareholdings of the government firm, in order to match the equilibrium distribution of the revenue from selling the quota. 

\begin{theorem}\label{equnderctisquota}
Let $\mathcal E=\{(X, \succ_{\omega}, P_\omega, e_\omega, \theta)_{\omega\in \Omega}, (Y_j)_{j\in J}, (m^{(j)})_{j\in J}, \mathcal{Z}(m)\}$ be a finite production economy with quota
and $(\bar{x},\bar{y},\bar{p})$ be a quota equilibrium.
Consider a finite production economy with quota 
$\mathcal{F}=\{(X, \succ_{\omega}, P_\omega, e_\omega, \tilde{\theta})_{\omega\in \Omega}, (Y_j)_{j\in J}, (\tilde{m}^{(j)})_{j\in J}, \mathcal{Z}(m)\}$: 
\begin{enumerate}
\item all private firm's quota are $0$, and the government firm's quota is $-m$,\footnote{Note that we have $\sum_{j\in J}\tilde{m}^{(j)}=\sum_{j\in J}m^{(j)}=m$.} i.e., $\mathcal{F}$ is a global quota economy;
\item $\tilde{\theta}(\omega)(j)=\theta(\omega)(j)$ for every private firm $j\in J$, and $\tilde{\theta}(\omega)(0)=\frac{\sum_{j\in J}\theta_{\omega j}\pi_{k}(\bar{p})\cdot m^{(j)}}{\pi_{k}(\bar{p})\cdot m}$.
\end{enumerate}
Then $(\bar{x},\bar{y},\bar{p})$ is a $\mathcal{Z}(m)$-compliant global quota equilibrium for $\mathcal{F}$. 
\end{theorem}
\begin{proof}
Agent $\omega$'s budget set $B_{\omega}^{m}(\bar{y},\bar{p})$ is
\[
\{z\in X_{\omega}: \bar{p}\cdot z\leq \bar{p}\cdot e(\omega)+\sum_{j\in J}\theta_{\omega j}(\bar{p}\cdot \bar{y}(j)+\pi_{k}(\bar{p})\cdot m^{(j)})\}. \nonumber
\]
As $\sum_{j\in J}\theta_{\omega j}\pi_{k}(\bar{p})\cdot m^{(j)}=\tilde{\theta}(\omega)(0)\pi_{k}(\bar{p})\cdot m$, the budget set $\tilde{B}_{\omega}^{m}(\bar{y},\bar{p})$ of agent $\omega$ in the economy $\mathcal{F}$ is the same as $B_{\omega}^{m}(\bar{y},\bar{p})$. 
Hence, $\bar{x}(\omega)$ is an element of the quota demand set $\tilde{D}^{m}_{\omega}(\bar{x},\bar{y},\bar{p})$ for all $\omega\in \Omega$ in the economy $\mathcal{F}$.
As $\bar{y}(j)\in S^{m}_{j}(\bar{p})$ for all $j\in J$, we have $\bar{y}(j)\in \argmax_{z\in Y_j}\bar{p}\cdot z$. Finally, as $\sum_{\omega\in \Omega}\bar{x}(\omega)-\sum_{\omega\in \Omega}e(\omega)-\sum_{j\in J}\bar{y}(j)\in \mathcal{Z}(m)$, we conclude that $(\bar{x},\bar{y},\bar{p})$ is a $\mathcal{Z}(m)$-compliant global quota equilibrium for $\mathcal{F}$.
\end{proof}

Note that the government firm's shareholdings 
in $\mathcal{F}$ in \cref{equnderctisquota} depend on the equilibrium price in $\mathcal{E}$, and the excess demand is a different function of price in the two economies; they agree only at the equilibrium price, which is the same in $\mathcal{E}$ and $\mathcal{F}$. 
Thus, different quota equilibria in $\mathcal{E}$ are global quota equilibria in possibly different global quota economies. 
However, when there is only one regulated commodity, the excess demand is the same function of price in $\mathcal{E}$ and $\mathcal{F}$, hence every equilibrium of a production economy with quota can be realized as a global quota equilibrium of a single global quota economy.

\begin{theorem}\label{capecoquota}
Let $\mathcal E=\{(X, \succ_{\omega}, P_\omega, e_\omega, \theta)_{\omega\in \Omega}, (Y_j)_{j\in J}, (m^{(j)})_{j\in J}, \mathcal{Z}(m)\}$ be a finite production economy with quota such that $k=1$.\footnote{There is a single regulated commodity. Hence, we have $m^{(j)}\in \NegReals$ for all $j\in J$.} Consider another finite production economy with quota 
$\mathcal{F}=\{(X, \succ_{\omega}, P_\omega, e_\omega, \tilde{\theta})_{\omega\in \Omega}, (Y_j)_{j\in J}, (\tilde{m}^{(j)})_{j\in J}, \mathcal{Z}(m)\}$ where:
\begin{enumerate}
\item all private firm's quota are $0$, and the government firm's quota is $-m$, i.e., $\mathcal{F}$ is a global quota economy;

\item $\tilde{\theta}(\omega)(j)=\theta(\omega)(j)$ for every private firm $j\in J$, and $\tilde{\theta}(\omega)(0)=\frac{\sum_{j\in J}\theta_{\omega j}m^{(j)}}{m}$. 
\end{enumerate}
Then every quota equilibrium in $\mathcal{E}$ is a global quota equilibrium in $\mathcal{F}$.
\end{theorem}
\begin{proof}
As $k=1$, $\sum_{j\in J}\theta_{\omega j}\pi_{k}(p)\cdot m^{(j)}=\tilde{\theta}(\omega)(0)\pi_{k}(p)\cdot m$ for $p\in \Delta$. By the same proof of \cref{equnderctisquota}, a quota equilibrium in $\mathcal{E}$ is a global quota equilibrium in $\mathcal{F}$.
\end{proof}

\subsection{Welfare Theorem for Production Economies with Quota}\label{secwelfare}
Since net $\mathrm{CO}_{2}$ emissions affect the temperature, emissions is an important factor to consider in comparing a house in Minneapolis to one in Miami. In this section, we focus on this specific type of externality and investigate the welfare properties of quota equilibrium. In particular, we assume that the only externality arises from the total net emission of the regulated commodities and 
establish that every quota equilibrium consumption-production pair is constrained Pareto optimal, i.e., it is Pareto optimal among all quota-compliant consumption-production pairs with the same total net emission of the regulated commodities.
As it is possible for a quota equilibrium consumption-production pair to Pareto dominate another quota equilibrium consumption-production pair with a different total net emission of the regulated commodities, a quota equilibrium consumption-production pair is, in general, not full Pareto optimal among all quota-compliant consumption-production pairs. Full Pareto optimality, however, can often be achieved if the government sets the right quota.

The total net pollution emission depend on agents' consumption and the aggregate production. In particular, for any $(x,y)\in \mathcal{A}\times Y$ (not necessarily quota-compliant), 
$C(x,y)=\pi_{k}\big(\sum_{\omega\in \Omega}e(\omega)+\sum_{j\in J}y(j)-\sum_{\omega\in \Omega}x(\omega)\big)$ is the total net emission of the regulated commodities, where $\pi_{k}$ is the projection map onto the first $k$ coordinates.
The externality in agents' preferences needs to be taken into account in defining Pareto domination. As the only externality arises from the total net pollution emission, an agent $\omega$'s global preference relation $\succ_{\omega}$ is a preference defined on $X_{\omega}\times G$, where $G=\pi_{k}\big(\sum_{\omega\in \Omega}e(\omega)+\bar{Y}-\bar{X}\big)$\footnote{$\bar{X}=\sum_{\omega\in \Omega}X_{\omega}$ is the aggregate consumption set.} denote the set of all possible total net emissions of the regulated commodities. 
We define \emph{Pareto domination} and \emph{full Pareto optimality} as: 
\begin{definition}
For two quota-compliant consumption-production pairs $(f, y), (f',y')\in \mathcal{A}\times Y$, we say $(f, y)$ \emph{Pareto dominates} $(f', y')$ if:
 \begin{itemize}
    \item for all $\omega\in \Omega$, $(f'(\omega), C(f',y'))\not\succ_{\omega} (f(\omega), C(f,y))$;
    \item there exists some $\omega_0\in \Omega$ such that $(f(\omega_0), C(f,y))\succ_{\omega_0}(f'(\omega_0), C(f',y'))$. 
\end{itemize}
A consumption-production pair $(g, h)$ \emph{strongly Pareto dominates} another consumption-production pair $(g',h')$ if $(g(\omega),C(g,h))\succ_{\omega} (g'(\omega),C(g',h'))$ for all $\omega\in \Omega$.
A consumption-production pair $(f, y)$ is \emph{(weakly) Pareto optimal} among $F\subset \mathcal{A}\times Y$ if no consumption-production pair in $F$ (strongly) Pareto dominates $(f, y)$. A consumption-production $(f, y)$ is \emph{(weakly) full Pareto optimal} if it is (weakly) Pareto optimal among all quota-compliant consumption-production pairs.
\end{definition}

It is too much to hope that quota equilibria are full Pareto optimal. After all, setting the quota to zero would result in an immediate return to a pre-industrial society. Therefore, we need to introduce the following weakened notion of Pareto optimality to better characterize the social welfare properties of quota equilibrium. 
\begin{definition}
A quota-compliant consumption-production pair $(f,y)\in \mathcal{A}\times Y$ is \emph{(weakly) constrained Pareto optimal} if $(f, y)$ is (weakly) Pareto optimal among all quota-compliant consumption-production pairs with the same total net emissions of the regulated commodities. 
\end{definition}
Our main result in this section indicates that any quota equilibrium consumption-production pair is constrained Pareto optimal.

\begin{customthm}{2}\label{fstwelfarequota}
Let $\mathcal E=\{(X, \succ_{\omega}, P_\omega, e_\omega, \theta)_{\omega\in \Omega}, (Y_j)_{j\in J}, (m^{(j)})_{j\in J}, \mathcal{Z}(m)\}$ be a finite production economy with quota, and the only externality arises from the total net emission of the first $k$ commodities. Suppose $\mathcal{Z}(m)_{n}=\{0\}$ for all $k<n\leq \ell$, i.e., we require non-free-disposal for non-regulated commodities at equilibrium.\footnote{Even in the context of \citet{ad54} and \citet{Mck59}, the first welfare theorem may fail if some coordinates of the equilibrium prices are negative or the equilibrium allows free disposal. Since we allow for possible negative equilibrium prices for non-regulated commodities, non-free-disposal for non-regulated commodities is necessary for the validity of this theorem. In fact, if we allow for free-disposal of electricity in \cref{exampleproperty}, then there is a quota equilibrium that is not constrained Pareto optimal.}
Let $(\bar{f}, \bar{y}, \bar{p})$ be a $\mathcal{Z}(m)$-compliant quota equilibrium. Then:
\begin{enumerate}
    \item $(\bar{f},\bar{y})$ is weakly constrained Pareto optimal;
    \item Suppose $P_{\omega}(\bar{f},\bar{y},\bar{p})$ is negatively transitive and locally non-satiated for all $\omega\in \Omega$. Then $(\bar{f},\bar{y})$ is constrained Pareto optimal.
\end{enumerate}
\end{customthm}
\begin{proof}
Suppose there exists some quota-compliant consumption-production pair $(\hat{f},\hat{y})$ such that $C(\hat{f},\hat{y})=C(\bar{f},\bar{y})$ and strongly  Pareto dominates $(\bar{f},\bar{y})$. Then, we have $\big(\hat{f}(\omega), \bar{f}(\omega)\big)\in P_{\omega}(\bar{f},\bar{y},\bar{p})$ for all $\omega\in \Omega$. As $(\bar{f}, \bar{y}, \bar{p})$ is a $\mathcal{Z}(m)$-compliant quota equilibrium, we have
\[
\bar{p}\cdot \hat{f}(\omega)&>\bar{p}\cdot e(\omega)+\sum_{j\in J}\theta_{\omega j}\big(\bar{p}\cdot \bar{y}(j)+\pi_{k}(\bar{p})\cdot m^{(j)}\big)\nonumber \\ 
&\geq  \bar{p}\cdot e(\omega)+\sum_{j\in J}\theta_{\omega j}\big(\bar{p}\cdot \hat{y}(j)+\pi_{k}(\bar{p})\cdot m^{(j)}\big)\nonumber
\]
for all $\omega\in \Omega$. Thus, we have
\[
\bar{p}\big(\sum_{\omega\in \Omega}\hat{f}(\omega)-\sum_{\omega\in \Omega}e(\omega)-\sum_{j\in J} \hat{y}(j)\big)>\pi_{k}(\bar{p})\cdot m. \nonumber
\]
As $\mathcal{Z}(m)_{n}=\{0\}$ for all $k<n\leq \ell$ and $(\hat{f},\hat{y})\in [C(\bar{f},\bar{y})]_{\mathrm{total}}$, \footnote{If we know that the equilibrium prices for non-regulated commodities are non-negative, then the conclusion of the theorem remains valid without assuming $\mathcal{Z}(m)_{n}=\{0\}$ for all $k<n\leq \ell$.}we have 
\[
\bar{p}\big(\sum_{\omega\in \Omega}\hat{f}(\omega)-\sum_{\omega\in \Omega}e(\omega)-\sum_{j\in J} \hat{y}(j)\big)=\pi_{k}(\bar{p})\cdot m, \nonumber
\]
which leads to a contradiction. Hence $(\bar{f},\bar{y})$ is weakly constrained Pareto optimal.

We now show that $(\bar{f},\bar{y})$ is constrained Pareto optimal if $P_{\omega}(\bar{f},\bar{y},\bar{p})$ is negatively transitive and locally non-satiated for all $\omega\in \Omega$.
Suppose there exists some quota-compliant consumption-production pair $(\hat{f},\hat{y})$ such that $C(\hat{f},\hat{y})=C(\bar{f},\bar{y})$ and Pareto dominates $(\bar{f},\bar{y})$. Then there exists some $\omega_0\in \Omega$ such that $\big(\hat{f}(\omega_0), \bar{f}(\omega_0)\big)\in P_{\omega_0}(\bar{f},\bar{y},\bar{p})$. 
As $(\bar{f}, \bar{y}, \bar{p})$ is a $\mathcal{Z}(m)$-compliant quota equilibrium, we have:
\[
\bar{p}\cdot \hat{f}(\omega_0)>\bar{p}\cdot e(\omega_0)+\sum_{j\in J}\theta_{\omega_{0} j}\big(\bar{p}\cdot \hat{y}(j)+\pi_{k}(\bar{p})\cdot m^{(j)}\big). \nonumber
\]
To complete the proof, we need the following key result: 
\begin{claim}\label{nolesshatf}
For every $\omega\in \Omega$, we have:
\[
\bar{p}\cdot \hat{f}(\omega)\geq \bar{p}\cdot e(\omega)+\sum_{j\in J}\theta_{\omega j}\big(\bar{p}\cdot \bar{y}(j)+\pi_{k}(\bar{p})\cdot m^{(j)}\big). \nonumber
\]
\end{claim}
\begin{proof}
Suppose there exists some $\omega_1\in \Omega$ such that 
\[
\bar{p}\cdot \hat{f}(\omega_{1})<\bar{p}\cdot e(\omega_{1})+\sum_{j\in J}\theta_{\omega_{1} j}\big(\bar{p}\cdot \bar{y}(j)+\pi_{k}(\bar{p})\cdot m^{(j)}\big). \nonumber
\]
As $P_{\omega_{1}}(\bar{f},\bar{y},\bar{p})$ is locally non-satiated, then there exists some $u\in X_{\omega_{1}}$ such that $\big(u, \hat{f}(\omega_{1})\big)\in P_{\omega_{1}}(\bar{f},\bar{y},\bar{p})$ and 
$\bar{p}\cdot u<\bar{p}\cdot e(\omega_{1})+\sum_{j\in J}\theta_{\omega_{1} j}\big(\bar{p}\cdot \bar{y}(j)+\pi_{k}(\bar{p})\cdot m^{(j)}\big)$.
As $P_{\omega_{1}}(\bar{f},\bar{y},\bar{p})$ is negatively transitive, we have $\big(u, \bar{f}(\omega_{1})\big)\in P_{\omega_{1}}(\bar{f},\bar{y},\bar{p})$. This leads to a contradiction since $(\bar{f}, \bar{y}, \bar{p})$ is a $\mathcal{Z}(m)$-compliant quota equilibrium. 
\end{proof}
By \cref{nolesshatf}, we have 
\[
\bar{p}\cdot \hat{f}(\omega)\geq \bar{p}\cdot e(\omega)+\sum_{j\in J}\theta_{\omega j}\big(\bar{p}\cdot \hat{y}(j)+\pi_{k}(\bar{p})\cdot m^{(j)}\big) \nonumber
\]
for all $\omega\in \Omega$. So we have $\bar{p}\big(\sum_{\omega\in \Omega}\hat{f}(\omega)-\sum_{\omega\in \Omega}e(\omega)-\sum_{j\in J} \hat{y}(j)\big)>\pi_{k}(\bar{p})\cdot m$. By the same argument as in the first paragraph, this leads to a contradiction which allows us to conclude that $(\bar{f},\bar{y})$ is constrained Pareto optimal. 
\end{proof}

\cref{fstwelfarequota} shows that, after the quota is set, constrained Pareto optimality of the equilibrium consumption-production pair is achieved without further intervention from the government. 
If in addition, we assume that agents cannot consume any of the regulated commodities,\footnote{\label{noconsumefoot} In the classical general equilibrium model developed \citet{ad54}, equilibrium assigns ownership which conveys the right to consume the commodity, but does not entail the obligation to eliminate it. In the regulation of pollution, a firm that accepts payment to take ownership of emissions must be obligated to eliminate those emissions, for example, by sequestering $\mathrm{CO}_{2}$. Moreover, individual consumers are technologically incapable of eliminating emissions, an industrial process available only to firms. Therefore, we assume agents cannot consume the regulated commodities.} 
then, since endowments are fixed, the total net emissions of the regulated commodities depend only on the production. The following result is similar to \cref{fstwelfarequota} except that the total net emissions of the regulated commodities are replaced by the total net production of the regulated commodities. 

\begin{customcor}{1}\label{welfareprod}
Let $\mathcal E=\{(X, \succ_{\omega}, P_\omega, e_\omega, \theta)_{\omega\in \Omega}, (Y_j)_{j\in J}, (m^{(j)})_{j\in J}, \mathcal{Z}(m)\}$ be a finite production economy with quota, and the only externality arises from the total net emission of the first $k$ commodities. Suppose  $\mathcal{Z}(m)_{n}=\{0\}$ for all $k<n\leq \ell$ and $\pi_{k}(X_{\omega})=\{0\}$ for all $\omega\in \Omega$. Let $(\bar{f}, \bar{y}, \bar{p})$ is a $\mathcal{Z}(m)$-compliant quota equilibrium. Then:
\begin{enumerate}
    \item $(\bar{f},\bar{y})$ is weakly constrained Pareto optimal, i.e., $(\bar{f},\bar{y})$ is weakly Pareto optimal among all quota-compliant production-consumption pairs with the same total production of the regulated commodities;
    \item Suppose $P_{\omega}(\bar{f},\bar{y},\bar{p})$ is negatively transitive and locally non-satiated for all $\omega\in \Omega$. Then $(\bar{f},\bar{y})$ is constrained Pareto optimal, i.e., $(\bar{f},\bar{y})$ is Pareto optimal among all quota-compliant production-consumption pairs with the same total production of the regulated commodities. 
\end{enumerate}
\end{customcor}

Since the total net emission of the regulated commodities is likely to affect agents' preferences, there may be a Pareto ranking among quota equilibrium with different total net emission of the regulated commodities. We conclude this section with the following example:

\begin{example}[continues=exampleproperty]\label{quotawelfareexample}
We investigate agents' welfare properties under different quota allocations between the government firm and private firms. 

We first consider the case where the private firm's quota is $0$. As we have shown in \cref{exampleproperty}, given $m\in (-2, 0)$, both agents' equilibrium consumption of electricity is 
$\frac{-m}{2}$. Hence, both agent's utility is given by $\frac{-m}{2}-\frac{m^{2}}{2}$. By taking the derivative, we see that $m=-\frac{1}{2}$ uniquely maximize both agents' utility. Thus, by setting the $\mathrm{CO}_{2}$ emission to be $\frac{1}{2}$, the resulting quota equilibrium consumption-production pair is full Pareto optimal, and allocates the same consumption to both agents. 

We now consider the case where the government firm's quota is $0$. As we have shown in \cref{exampleproperty}, given $m\in (-2, 0)$, the first agent's equilibrium electricity consumption is $-m$ while the second agent's equilibrium electricity consumption is $0$; all of the consumption goes to the first agent. However, the second agent's utility is adversely affected by the $\mathrm{CO}_{2}$ emissions. In particular, the first agent's utility is given by $(-m)-\frac{m^{2}}{2}$ and the second agent's utility is given by $-\frac{m^{2}}{2}$. The second agent's utility is an increasing function for $m\in [-2, 0]$ while the first agent's utility is an increasing function for $m\in [-2,-1)$ and a decreasing function for $m\in [-1, 0]$. Thus, for every $m\in [-1, 0]$, the resulting $\mathcal{Z}(m)$-compliant quota equilibrium consumption production pair is full Pareto optimal. 
\end{example}

\section{Production Economy with Tax}\label{sectaxequil}
An alternative to setting quotas is to set tax rates on pollution emissions.
In this section, we present a general equilibrium model that incorporates a tax regulatory scheme on net pollution emissions, and study its connection with quota equilibrium.  
While quota equilibrium always exists and is effective in limiting the total net pollution emission, we demonstrate in \cref{equndertaxisquota}, \cref{equnderqtistax} and \cref{taxequipbexample} that: 
\begin{enumerate}
    \item Every quota equilibrium can be realized as an emission tax equilibrium and vice versa; 
    \item Given specific tax rates, there may be no emission tax equilibrium consistent with those tax rates;
    \item Even if an emission tax equilibrium exists for a given tax rate, there may be multiple equilibria and there is no guarantee that emissions will lie under a prespecified level of total net pollution emissions for every equilibrium. 
\end{enumerate}
On the other hand, \cref{nonlinearexample} shows that it might be possible to limit the total net pollution emission and achieve full Pareto optimality through an emission tax if there is a one-to-one correspondence between tax rates and pollution emissions. Finally, in \cref{secfuel}, we reconsider \cref{taxequipbexample} but instead place a fuel tax on the input of production. We discover that, in the context of this example, fuel tax equilibrium is Pareto dominated by the emission tax equilibrium given a total net $\mathrm{CO}_{2}$ emission level. 
We now give a rigorous formulation of a general equilibrium model that incorporates tax on emission of pollution:
\begin{definition}\label{defecotax}
A finite production economy with emission tax
\[
\mathcal{F}\equiv \{(X, \succ_{\omega}, P_\omega, e_\omega, \theta)_{\omega\in \Omega}, (Y_j)_{j\in J},  \mathcal{V}, t, \lambda\} \nonumber
\] 
is a list  such that 
\begin{enumerate}[(i)]
\item All firms are private firms. Consumption sets, preferences, endowments, production sets, and agents' shares of private firms are defined the same as in \cref{defmeasureeco}. Hence, $(X, \succ_{\omega}, P_\omega, e_\omega, \theta)_{\omega\in \Omega}$ and $(Y_j)_{j\in J}$ and well-defined;
\item $\lambda:\Omega\to \NNReals$ is the government's rebate share such that $\sum_{\omega\in \Omega}\lambda(\omega)=1$. The government rebates its revenue from tax to agents according to $\lambda$; 
\item The compliance region $\mathcal{V}$ takes the form of $\prod_{n\leq \ell}\mathcal{V}_{n}$ where $\mathcal{V}_{n}=\NegReals$ for all $n\leq k$ and $\mathcal{V}_{n}$ is either $\{0\}$ or $\NegReals$ for all $n>k$;
\item $t\in \pi_{k}(\Delta)$ is an emission tax rate on the net emission of the first $k$ commodities.
\end{enumerate}
\end{definition}

\begin{remark}
The features of our general equilibrium model with emission tax are: 
\begin{itemize}
\item Our model focuses on emission taxes. An alternative formulation is to impose fuel taxes on commodities that would generate pollution as a byproduct in the production process. 
For example, a tax on gasoline is an fuel tax: the consumer pays the equilibrium price and the tax to the seller, who keeps the equilibrium price and remits the tax to the government. Fuel taxes have a significant advantage: it is much easier to collect a gasoline tax at the service station than it is to measure the tailpipe emissions of each vehicle. However, fuel taxes may be inefficient because they tax all uses of a commodity, irrespective of emissions generated. For example, crude oil and natural gas may be burned, releasing a large quantity of $\mathrm{CO}_{2}$ into the atmosphere, or may be used in the production of chemicals such as fertilizer, releasing a much smaller quantity of $\mathrm{CO}_{2}$. Indeed, carbon sequestration relies on using some of the electricity generated by burning coal in order to reduce $\mathrm{CO}_{2}$ emissions. A fuel tax on coal discourages sequestration by taxing the portion of the coal devoted to electricity generation needed for sequestration, as in example \cref{cbtaxequil}. As a consequence, a fuel tax may be inefficient compared to the corresponding emission tax. One might hope that, because the fuel tax model allows all equilibrium prices to adjust, fuel tax equilibrium might always exist. However, \cref{addonnoexst} shows that fuel tax equilibrium may not exist for certain tax rate;

\item The government rebate share $\lambda$ is parallel to agents' shareholdings of the government firm in production economies with quota, which are both determined by the government. Moreover, the property rights in a production economy with tax coincide with the property rights in the corresponding global quota economy, since all rights are vested to the government, and the profits are distributed to agents according to the rebate scheme;

\item As the first $k$ coordinates of the compliance region $\mathcal{V}$ are $\NegReals$, we allow for arbitrary net emissions for the regulated commodities, but charge a tax on the emissions of these commodities. As in production economies with quota, we allow for either free disposal or non-free-disposal of the non-regulated commodities. 
\end{itemize}
\end{remark}
Recall that $C(x,y)=\pi_{k}\big(\sum_{\omega\in \Omega}e(\omega)+\sum_{j\in J}y(j)-\sum_{\omega\in \Omega}x(\omega)\big)$ is the total net emission of the regulated commodities. 
For every $\omega\in \Omega$, $p\in \Delta$ and $(x,y)\in \mathcal{A}\times Y$, the \emph{emission tax budget set} $B_{\omega}^{t}(x, y, p)$ is defined to be:
\[
\{z\in X_{\omega}: p\cdot z \leq p\cdot e(\omega) + \sum_{j\in J}\theta_{\omega j}p\cdot y(j)+\lambda(\omega)t\cdot C(x,y)\}.\nonumber
\]
So an agent's budget consists of the value of her endowment, her dividends from firms and her rebate from the government's emission tax revenue. 
For each $\omega\in \Omega$ and $(x,y,p)\in \mathcal{A}\times Y\times \Delta$, the \emph{emission tax demand set} $D^{t}_{\omega}(x,y,p)$ consists of all elements in the emission tax budget set $B_{\omega}^{t}(x, y, p)$ that maximize the agent's preference given $(x,y,p)$. 
In particular, $D^{t}_{\omega}(x,y,p)$ is defined as:
\[
\{z \in B^{t}_{\omega}(x,y,p): w \succ_{x,y,\omega,p} z\implies w\not\in B_{\omega}^{t}(x,y,p)\}.\nonumber
\]
The equilibrium notion for finite production economies with emission tax is: 
\begin{definition}\label{deftaxequil}
Let $\mathcal{F}=\{(X, \succ_{\omega}, P_\omega, e_\omega, \theta)_{\omega\in \Omega}, (Y_j)_{j\in J}, \mathcal{V}, t, \lambda\}$ be a finite production economy with emission tax. 
A $\mathcal{V}$-compliant emission tax equilibrium is a triple $(\bar{x}, \bar{y}, \bar{p})\in \mathcal{A}\times Y\times \Delta$ such that the following conditions are satisfied:
\begin{enumerate}
\item $\pi_{k}(\bar{p})=-t$;
\item $\bar{x}(\omega)\in D^{t}_{\omega}(\bar{x},\bar{y},\bar{p})$ for all $\omega\in \Omega$;
\item $\bar{y}(j)\in S_{j}(\bar{p})\equiv \argmax_{z\in Y_{j}}\bar{p}\cdot z$ for all $j\in J$. So every firm is profit maximizing given the price $\bar{p}$;
\item $\sum_{\omega\in \Omega}\bar{x}(\omega)-\sum_{\omega\in \Omega}e(\omega)-\sum_{j\in J}\bar{y}(j)\in \mathcal{V}$.
\end{enumerate}
\end{definition}

Once the government sets the tax rate and the rebate share, the total net emission of the regulated commodities is determined endogenously.
We now study the connection between quota equilibrium and emission tax equilibrium. In particular, we show that every quota equilibrium can be realized as a tax equilibrium and vice versa.

\begin{customthm}
{3}\label{equndertaxisquota}
Let $\mathcal E=\{(X, \succ_{\omega}, P_\omega, e_\omega, \theta)_{\omega\in \Omega}, (m^{(j)})_{j\in J}, (Y_j)_{j\in J},  \mathcal{Z}(m)\}$ be a finite production economy with quota as in \cref{defmeasureeco}. Let $(\bar{x},\bar{y},\bar{p})$ be a $\mathcal{Z}(m)$-compliant quota equilibrium of $\mathcal{E}$. Let $\mathcal{F}=\{(X, \succ_{\omega}, P_\omega, e_\omega, \theta)_{\omega\in \Omega}, (Y_j)_{j\in J'}, \mathcal{V}, t, \lambda\}$ be a finite production economy with emission tax such that:
\begin{enumerate}
    \item  $\mathcal{V}_{n}=\NegReals$ for all $n<k$ and $\mathcal{V}_{n}=\mathcal{Z}(m)_{n}$ for all $n\geq k$;
    \item $J'$ is the set of private firms in $J$;
    \item $t=-\pi_{k}(\bar{p})$;
    \item  $\lambda(\omega)=\frac{\sum_{j\in J}\theta_{\omega j}\pi_{k}(\bar{p})\cdot m^{(j)}}{\pi_{k}(\bar{p})\cdot m}$.
\end{enumerate}
Then, $(\bar{x},\bar{y},\bar{p})$ is a $\mathcal{V}$-compliant emission tax equilibrium of $\mathcal{F}$. Suppose, in addition, $k=1$, i.e., there is one regulated commodity. Then $\lambda(\omega)=\frac{\sum_{j\in J}\theta_{\omega j}m^{(j)}}{m}$ is independent of the price. 
\end{customthm}
\begin{proof}
We first consider the special case where $\mathcal{E}$ is a global quota economy, i.e., $m^{(j)}=0$ for every private firm $j\in J$. As $(\bar{x},\bar{y},\bar{p})$ is a $\mathcal{Z}(m)$-compliant quota equilibrium, 
we have $\bar{y}(j)\in S_{j}(\bar{p})$ for all $j\in J$. As $\mathcal{V}\supset \mathcal{Z}(m)$, we have 
$\sum_{\omega\in \Omega}\bar{x}(\omega)-\sum_{\omega\in \Omega}e(\omega)-\sum_{j\in J}\bar{y}(j)\in \mathcal{V}$.
As $t=-\pi_{k}(\bar{p})$ and $C(\bar{x},\bar{y})=-m$, 
we have $B_{\omega}^{m}(\bar{y},\bar{p})=B_{\omega}^{t}(\bar{x},\bar{y},\bar{p})$ for all $\omega\in \Omega$. As $\bar{x}(\omega)\in D_{\omega}^{m}(\bar{x},\bar{y},\bar{p})$ for all $\omega\in \Omega$, we have $\bar{x}(\omega)\in D_{\omega}^{t}(\bar{x},\bar{y},\bar{p})$ for all $\omega\in \Omega$.
Hence, $(\bar{x},\bar{y},\bar{p})$ is a $\mathcal{V}$-compliant emission tax equilibrium for the economy $\mathcal{F}$. 

For general production economies with quota, the conclusions follow from \cref{equnderctisquota} and \cref{capecoquota}.
\end{proof}

As in \cref{equnderctisquota}, the government's rebate share depends on the equilibrium price in $\mathcal{E}$, and the excess demand is a different function of price in the two economies; they agree only at the equilibrium price, which is the same in $\mathcal{E}$ and $\mathcal{F}$. 
Thus, different quota equilibria in $\mathcal{E}$ are emission tax equilibria in possibly different production economies with emission tax. However, when there is only one regulated commodity, the excess demand is the same function of price in $\mathcal{E}$ and $\mathcal{F}$, and the government's rebate share is independent of the equilibrium price in $\mathcal{E}$.  

Conversely, every emission tax equilibrium is a global quota equilibrium with the quota being the negative of the total net emission of the regulated commodities in the emission tax equilibrium:

\begin{customthm}
{4}\label{equnderqtistax}
Let $\mathcal F=\{(X, \succ_{\omega}, P_\omega, e_\omega, \theta)_{\omega\in \Omega}, (Y_j)_{j\in J}, \mathcal{V}, t, \lambda\}$ be a finite production economy with emission tax and $(\bar{x},\bar{y},\bar{p})$ be a $\mathcal{V}$-compliant emission tax equilibrium. Let $\mathcal{E}=\{(X, \succ_{\omega}, P_\omega, e_\omega, \tilde{\theta})_{\omega\in \Omega}, (m^{(j)})_{j\in J}, (Y_j)_{j\in J'}, \mathcal{Z}(m)\}$ be a global quota economy where:
\begin{enumerate}[leftmargin=*]
\item The set $J'$ is the union of $J$ and the government firm, denoted as firm $0$, with the production set $Y_{0}=\{0\}$;
\item $m=-C(\bar{x},\bar{y})$ and $\mathcal{Z}(m)_{n}=\mathcal{V}_{n}$ for all $n\geq k$;
\item For every $\omega\in \Omega$, we have $\tilde{\theta}(\omega)(j)=\theta(\omega)(j)$ for all $j\in J$, and $\tilde{\theta}(\omega)(0)=\lambda(\omega)$, i.e., the agents' shareholdings of the government firm are the same as the government's rebate share. 
\end{enumerate}
Then $(\bar{x},\bar{y},\bar{p})$ is a $\mathcal{Z}(m)$-compliant global quota equilibrium for $\mathcal{E}$. 
\end{customthm}
\begin{proof}
As $(\bar{x},\bar{y},\bar{p})$ is a $\mathcal{V}$-compliant emission tax equilibrium, we have $-C(\bar{x},\bar{y})\leq 0$. By choosing $m=-C(\bar{x},\bar{y})$, we have $\sum_{\omega\in \Omega}\bar{x}(\omega)-\sum_{\omega\in \Omega}e(\omega)-\sum_{j\in J}\bar{y}(j)\in \mathcal{Z}(m)$. As $\bar{y}(j)\in S_{j}^{m}(\bar{p})$ for all $j\in J$, 
we have $\bar{y}(j)\in S_{j}(\bar{p})$ for all $j\in J$. As $t=-\pi_{k}(\bar{p})$, $m=-C(\bar{x},\bar{y})$, $\tilde{\theta}(\omega)(0)=\lambda(\omega)$ for all $\omega\in \Omega$, and $\mathcal{E}$ is a global quota economy, we have $B_{\omega}^{m}(\bar{y},\bar{p})=B_{\omega}^{t}(\bar{x},\bar{y},\bar{p})$ for all $\omega\in \Omega$. Hence, we have $\bar{x}(\omega)\in D_{\omega}^{m}(\bar{x},\bar{y},\bar{p})$ for all $\omega\in \Omega$. So $(\bar{x},\bar{y},\bar{p})$ is a $\mathcal{Z}(m)$-compliant global quota equilibrium of $\mathcal{E}$. 
\end{proof}

\cref{equndertaxisquota} and \cref{equnderqtistax} show that a quota equilibrium can be realized as an emission tax equilibrium and vice versa. Hence, as we have pointed out in the introduction, these results provide rigorous arguments on the interchangeability of emission taxes and quotas. The existence and welfare properties of emission tax equilibrium, along with its comparison to quota equilibrium, will be studied through examples in the next section.

\subsection{Examples of Production Economies with Emission Tax}\label{secexample}
\cref{equndertaxisquota} and \cref{equnderqtistax} do not address the existence of emission tax equilibrium. 
In the next example, we show that emission tax equilibrium need not exist for certain tax rates, and, in some cases, it is impossible to limit emissions under a prespecified total net pollution emission level by setting tax rate.

\begin{customexp}{2}\label{taxequipbexample}\normalfont
Let $\mathcal{F}=\{(X, \succ_{\omega}, P_{\omega}, e_{\omega}, \theta_{\omega})_{\omega\in \Omega}, (Y_j)_{j\in J}, \mathcal{V}, t, \lambda\}$ be a finite production economy with emission tax: 
\begin{enumerate}
    \item The economy $\mathcal{E}$ has three commodities $\mathrm{CO}_{2}$, coal and electricity, which we denote by $c_1$, $c_2$ and $c_3$;
    \item There is a single agent with consumption set $X=\{0\}\times \NNReals^{2}$ and endowment $e=(0,1,0)$. Given the total net emission $v$ of $\mathrm{CO}_{2}$, the utility function $u_{v}(c_1,c_2,c_3)=c_3-\frac{v^{2}}{2}$;
    \item There are two producers with production sets $Y_1=\{(r,-r,r): r\in \NNReals\}$ and $Y_2=\{(-2r,0,-r): r\in \NNReals\}$. So the first producer has the production technology to burn $r$ units of coal to generate $r$ units of electricity and $r$ units of $\mathrm{CO}_{2}$ as byproduct. The second producer has the production technology to use $r$ units of electricity to sequester $2r$ units of $\mathrm{CO}_{2}$;
    \item The compliance region $\mathcal{V}=\NegReals^{3}$;
    \item Since there is only one agent, we have $\theta_{\omega}(1)=\theta_{\omega}(2)=\lambda(\omega)=1$.
\end{enumerate}

\textbf{Existence of Emission Tax Equilibrium}: We consider the existence of emission tax equilibrium under different emission tax rates.  

\begin{claim}\label{eqestclaim}
There is a $\mathcal{V}$-compliant emission tax equilibrium for emission if and only if the tax rate $t\leq \frac{1}{4}$.
\end{claim}

The proof of \cref{eqestclaim} is presented in \cref{appendixexp2}. \cref{eqestclaim} shows that emission tax equilibrium might not exist for specific emission tax rates.

\noindent \textbf{Limiting $\mathrm{CO}_{2}$ Emission via Emission Tax}: We now investigate whether setting an emission tax rate can ensure the total net $\mathrm{CO}_{2}$ emission is under a pre-specified level.
\begin{claim}\label{wrongtax}
The total net emission of $\mathrm{CO}_{2}$ at any $\mathcal{V}$-compliant emission tax equilibrium is $1$ if the emission tax rate $t<\frac{1}{4}$.  
\end{claim}

The proof of \cref{wrongtax} is presented in \cref{appendixexp2}.
By \cref{wrongtax}, if the government sets the emission tax rate to be less than $\frac{1}{4}$, then the total net $\mathrm{CO}_{2}$ emission is $1$ unit. 
We now consider the case where the emission tax rate is $\frac{1}{4}$. Note that, when the emission tax rate is $\frac{1}{4}$, the equilibrium price must be $(-\frac{1}{4},\frac{1}{4},\frac{1}{2})$. 

\begin{claim}\label{taxquatunique}
Let the emission tax rate be $\frac{1}{4}$ and $0\leq v\leq 1$. Then there exists a $\mathcal{V}$-compliant emission tax equilibrium such that the total net emission of $\mathrm{CO}_{2}$ is $v$.
\end{claim}

The proof of \cref{taxquatunique} is presented in \cref{appendixexp2}.
By \cref{taxquatunique}, there are multiple $\mathcal{V}$-compliant emission tax equilibrium associated with the emission tax rate $\frac{1}{4}$. In fact, as illustrated in Figure ~\ref{linearfigure}, for any prespecified total net $\mathrm{CO}_{2}$ emission $0\leq v\leq 1$, there is a $\mathcal{V}$-compliant emission tax equilibrium with emission tax rate $\frac{1}{4}$ with total net $\mathrm{CO}_{2} $emission being $v$. Combining with \cref{wrongtax}, we conclude that, in this example, it is impossible to set an emission tax rate to ensure that the $\mathrm{CO}_{2}$ emission is strictly less than $1$.\footnote{As we will see shortly, the full Pareto optimal $\mathrm{CO}_{2}$ emission level is $\frac{1}{2}$.}

\begin{figure}[h]
\caption*{Emission Tax with Linear Technology (\cref{taxequipbexample})}
\centering
\includegraphics[width=0.85\textwidth]{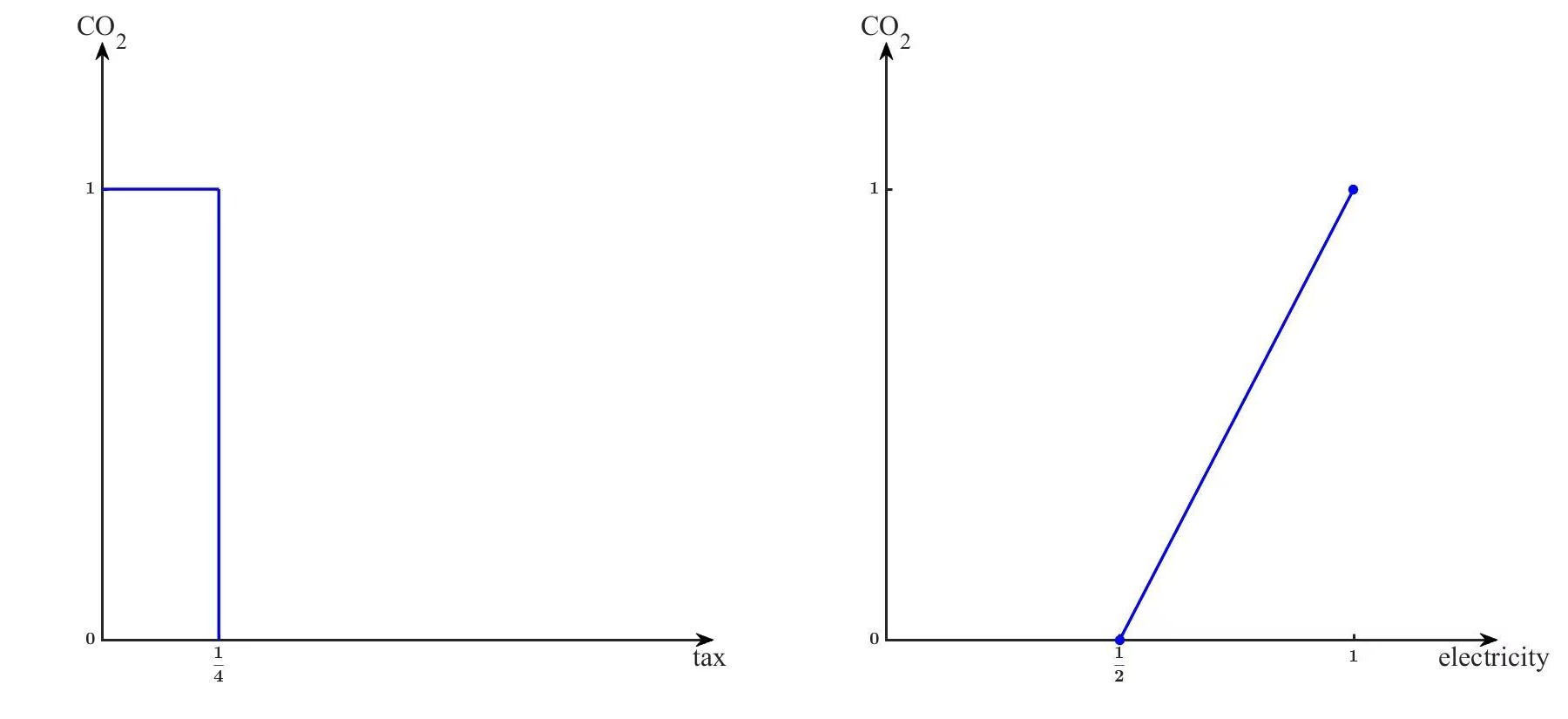}
\caption{The left figure plots the equilibrium total net $\mathrm{CO}_{2}$ emission as a correspondence of the tax rate $t\in [0,\frac{1}{4}]$. The right figure plots the 
electricity consumption/$\mathrm{CO}_{2}$ emissions pairs for emission tax equilibrium corresponding to $t=\frac{1}{4}$.}
\label{linearfigure}
\end{figure}

\textbf{The Welfare Properties of Emission Tax Equilibrium and Comparison with Quota Equilibrium}: We shall show that, in this example, it is impossible to ensure full Pareto optimality by setting an emission tax rate. However, we can ensure full Pareto optimality by carefully choosing the right quota. 

We start by considering the welfare properties of $\mathcal{V}$-compliant emission tax equilibria. 
There is a Pareto ranking among $\mathcal{V}$-compliant emission tax equilibria arising from the externality. 
By taking the derivative, the total net emission of $\mathrm{CO}_{2}$ that maximizes the agent's utility is $\hat{v}=\frac{1}{2}$. 
As a result, the $\mathcal{V}$-compliant emission tax equilibrium consumption-production pair $(\hat{x}_{\frac{1}{2}}, \hat{y}_{\frac{1}{2}})$\footnote{In particular, we have $\hat{x}_{\frac{1}{2}}=(0,0,\frac{3}{4})$, $\hat{y}_{\frac{1}{4}}=\big((1,-1,1), (-\frac{1}{2},0,-\frac{1}{4})\big)$ and $\hat{p}=(-\frac{1}{4}, \frac{1}{4},\frac{1}{2})$.} with emission tax rate $\frac{1}{4}$ Pareto dominates all other $\mathcal{V}$-compliant emission tax equilibrium consumption-production pairs. Thus, if the government sets the emission tax rate to be less than $\frac{1}{4}$, then total net $\mathrm{CO}_{2}$ emission of the resulting $\mathcal{V}$-compliant emission tax equilibrium is $1$, hence the resulting equilibrium consumption-production pair is Pareto dominated. 
On the other hand, there are multiple $\mathcal{V}$-compliant emission tax equilibrium with the emission tax rate $\frac{1}{4}$, but only one of these equilibrium consumption-production pairs is Pareto optimal. So it is impossible to guarantee full Pareto optimality by setting an emission tax rate on $\mathrm{CO}_{2}$.\footnote{In this example, the emission tax on $\mathrm{CO}_{2}$ can be viewed as the Pigouvian tax on electricity from the perspective of \citet{sandmo75}. The emission tax rate $\frac{1}{4}$ would be called the optimal Pigouvian tax rate since it is the emission tax rate associated with the unique Pareto optimal emission tax equilibrium. However, the emission tax rate $\frac{1}{4}$ is associated with multiple emission tax equilibria, and it is impossible to guarantee full Pareto optimality by setting the emission tax rate to be $\frac{1}{4}$.}

We now consider an associated finite production economy with quota $\mathcal{E}$. In particular, $\mathcal{E}=\{(X, \succ_{\omega}, P_{\omega}, e_{\omega}, \tilde{\theta}_{\omega})_{\omega\in \Omega}, (Y_j)_{j\in J'}, (m^{(j)})_{j\in J'}, \mathcal{Z}(m)\}$ is such that:
\begin{enumerate}
\item $J'$ is the union of $J$ and the government firm, denoted as firm $0$, with the production set $Y_0=\{0\}$;
\item $m$ is the negative of the quota on $\mathrm{CO}_{2}$ net emission, and $\mathcal{Z}(m)=\{m\}\times \NegReals^{2}$;
\item Since there is only one agent, we have $\tilde{\theta}_{\omega}(0)=\tilde{\theta}_{\omega}(1)=\tilde{\theta}_{\omega}(2)=1$;
\item $\mathcal{E}$ is a global quota economy. That is, we have $m^{(0)}=m$.\footnote{As there is only one agent, the quota allocation between the government and private firms has no impact on the equilibrium outcome}
\end{enumerate}
By \cref{taxquatunique} and \cref{equnderqtistax}, $\mathcal{E}$ has a $\mathcal{Z}(m)$-compliant quota equilibrium for all $m\in [-1,0]$.

By the definition of quota equilibrium, the total net emission of $\mathrm{CO}_{2}$ equals $-m$ for all $\mathcal{Z}(m)$-compliant quota equilibrium. Thus, unlike emission tax equilibrium, quota equilibrium always exists, and the quota can be chosen to ensure that the total net $\mathrm{CO}_{2}$ emission of every quota equilibrium equals the prespecified level. Moreover, by \cref{fstwelfarequota}, every quota equilibrium consumption-production pair is constrained Pareto optimal. That is, every quota equilibrium consumption-production pair is Pareto optimal among all quota-compliant consumption-production pairs with the same total net $\mathrm{CO}_{2}$ emission. Finally, as $(\hat{x}_{\frac{1}{2}}, \hat{y}_{\frac{1}{2}})$ associated with the tax rate $\frac{1}{4}$ is the unique Pareto optimal consumption-production pairs among all $\mathcal{V}$-compliant tax equilibrium, by \cref{equndertaxisquota}, $(\hat{x}_{\frac{1}{2}},\hat{y}_{\frac{1}{2}},\hat{p})$ is the only $\mathcal{Z}(-\frac{1}{2})$-compliant quota equilibrium of $\mathcal{E}$.
Thus, once the government sets the quota on $\mathrm{CO}_{2}$ emission to be $\frac{1}{2}$, full Pareto optimality of the equilibrium consumption-production pair is achieved without further intervention from the government.  
\end{customexp}

\cref{taxequipbexample} shows that emission tax equilibrium may not exist for certain emission tax rates, and setting an emission tax rate
does not ensure that the total net $\mathrm{CO}_{2}$ emission will be under a pre-specified level.
However, as the next example illustrates, if there exists a one-to-one correspondence between emission tax rates and total net $\mathrm{CO}_{2}$ emissions, one can not only limit the total net $\mathrm{CO}_{2}$ emission under a pre-specified level but also guarantee full Pareto optimality via emission tax rate.

\begin{customexp}
{3}\label{nonlinearexample}\normalfont
Let $\mathcal F$ be the finite production economy with emission tax as in \cref{taxequipbexample} except that the second firm's production set is now given by 
\[
Y_{2}=\{(-a,0,-r^{2}): (r\in \NNReals)\wedge (0\leq a\leq 2r)\}.\nonumber
\]
The second firm has the production technology to sequester $\mathrm{CO}_{2}$ using electricity, and the marginal quantity of electricity needed to sequester an additional unit of $\mathrm{CO}_{2}$ increases as the amount of $\mathrm{CO}_{2}$ that has been sequestered increases. This is qualitatively what we should expect: it is much easier to capture the first ten percent of the emissions of a given plant than the last ten percent.

\textbf{Existence and Properties of Emission Tax Equilibrium}: We consider the existence of emission tax equilibrium under different tax rates.

\begin{claim}\label{uniqueequilclaim}
There is a unique $\mathcal{V}$-compliant emission tax equilibrium for all emission tax rate $t\leq \frac{1}{4}$. There is no $\mathcal{V}$-compliant emission tax equilibrium for emission tax rate $t>\frac{1}{4}$.
\end{claim}

\begin{figure}[h]
\caption*{Emission Tax with Non-Linear Technology (\cref{nonlinearexample})}
\centering
\includegraphics[width=0.85\textwidth]{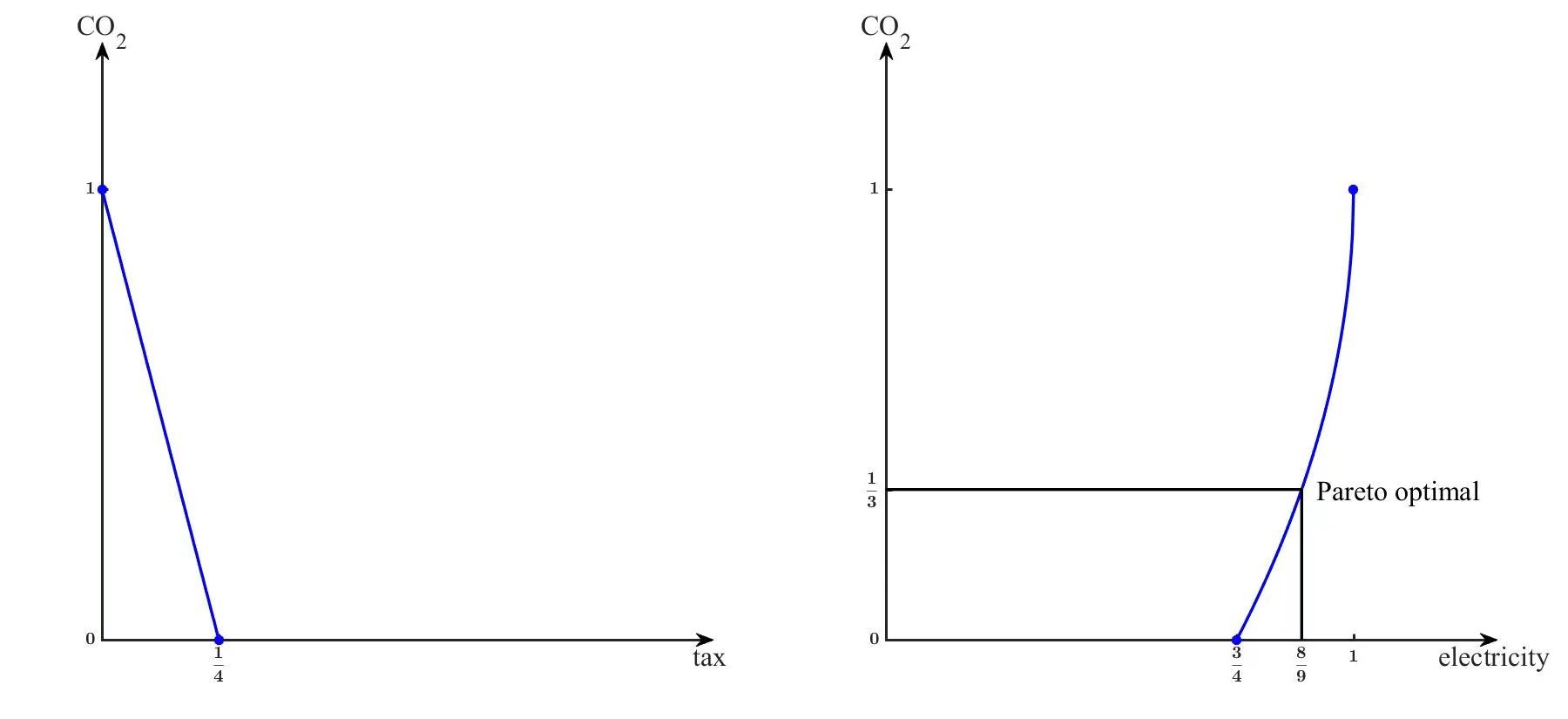}
\caption{The left figure plots the equilibrium total net $\mathrm{CO}_{2}$ emission as a function of the tax rate $t\in [0,\frac{1}{4}]$. The right figure plots the 
electricity consumption/$\mathrm{CO}_{2}$ emissions pairs for emission tax equilibrium for $t\in [0,\frac{1}{4}]$.}
\label{nlfigure}
\end{figure}

The proof of \cref{uniqueequilclaim} is presented in \cref{appendixexp3}. As illustrated in Figure ~\ref{nlfigure}, there is a one-to-one correspondence between the emission tax rate and the total net $\mathrm{CO}_{2}$ emission. 
Hence, although emission tax equilibrium does not exist for emission tax rate $t>\frac{1}{4}$, one can achieve any pre-specified total net $\mathrm{CO}_{2}$ emission via emission tax rate alone. In particular, by \cref{uniqueequilclaim}, given a tax rate $t\leq \frac{1}{4}$, the total net $\mathrm{CO}_{2}$ emission is $1-4t$. Hence, the government can limit the total net $\mathrm{CO}_{2}$ emission under any pre-specifed level $v\leq 1$ by setting the emission tax rate to be no less than $\frac{1-v}{4}$.

\textbf{The Welfare Properties of Emission Tax Equilibrium and Comparison with Quota Equilibrium}: We shall show that, in this example, one can ensure full Pareto optimality through either setting an emission tax rate or setting a quota.

As in \cref{taxequipbexample}, 
there is a Pareto ranking among $\mathcal{V}$-compliant emission tax equilibrium consumption-production pairs arising from the externality. 
By \cref{uniqueequilclaim}, the agent's utility at a $\mathcal{V}$-compliant emission tax equilibrium, as a function of the tax rate, is given by 
$(1-4t^{2})-\frac{(1-4t)^{2}}{2}$. By taking the derivative, 
the agent's utility maximizes uniquely at $t=\frac{1}{6}$ and the utility is $\frac{15}{18}$. Thus, the $\mathcal{V}$-compliant emission tax equilibrium consumption-production pair with emission tax rate $\frac{1}{6}$ Pareto dominates all other $\mathcal{V}$-compliant emission tax equilibrium consumption-production pairs. By \cref{uniqueequilclaim}, there exists a unique $\mathcal{V}$-compliant emission tax equilibrium with emission tax rate $\frac{1}{6}$. Hence, full Pareto optimality can be achieved by setting the emission tax rate to be $\frac{1}{6}$. 

We now consider the associated finite production economy with quota. Let $\mathcal{E}$ be a finite production economy with quota as in \cref{taxequipbexample} except that the second firm's production set is $Y_2=\{(-a,0,-r^{2}): (r\in \NNReals)\wedge (0\leq a\leq 2r)\}$. 
By \cref{uniqueequilclaim} and \cref{equnderqtistax}, $\mathcal{E}$ has a $\mathcal{Z}(m)$-compliant quota equilibrium for all $m\in [-1,0]$. 
Moreover, as there exists a one-to-one correspondence between tax rate and total net $\mathrm{CO}_{2}$ emission, it follows from \cref{equndertaxisquota} that $(\bar{x}_{\frac{1}{6}},\bar{y}_{\frac{1}{6}},\bar{p}_{\frac{1}{6}})$ is the only $\mathcal{Z}(-\frac{1}{3})$-compliant quota equilibrium. Thus, once the government sets the quota on $\mathrm{CO}_{2}$ to be $\frac{1}{3}$, full Pareto optimality of the equilibrium consumption-production pair is achieved without further intervention from the government. 
\end{customexp}

\subsection{Fuel Taxes}\label{secfuel}
We have so far focused on the emission taxes on regulated commodities instead of add-on taxes. 
In this section, we consider economies with three commodities: $\mathrm{CO}_{2}$, coal and electricity, as in \cref{secexample}. 
However, instead of placing a tax on $\mathrm{CO}_{2}$ emissions, we place a tax on the total utilization of coal as a consumption good and as an input to production. Moreover, our tax on coal is an add-on to the equilibrium price. In other words, in order to buy coal, one must pay the equilibrium price plus the tax.\footnote{Our ``tax" could be either positive or negative; in the latter case, it serves as a subsidy of the commodity.} 
In the next example, we find that, for a given total net $\mathrm{CO}_{2}$ emission level, the corresponding fuel tax equilibrium is Pareto dominated by the corresponding emission tax equilibrium.

\begin{customexp}
{4}\label{cbtaxequil}\normalfont
Let $\mathcal{F}=\{(X, \succ_{\omega}, P_{\omega}, e_{\omega}, \theta_{\omega})_{\omega\in \Omega}, (Y_j)_{j\in J}, \mathcal{V}, t, \lambda\}$ be a finite production economy as in \cref{taxequipbexample} but with a fuel tax on coal.\footnote{In this example, coal is not consumed, and is solely used as an input to production.} In particular, 
let $t\geq 0$ be the tax rate on coal. Through normalization, for a price vector $p$, we require that $t+\sum_{k=1}^{\ell}|p_{k}|=1$, that is, $(t,p)\in \Delta$. Note that only the first firm uses coal as an input in production, and the agent does not consume coal.  
For every $(t,p)\in \Delta$ and $y\in Y$, the agent's \emph{fuel tax budget set} $B^{c}_{\omega}(y,t,p)$ is defined to be:
\[
\{z\in X: p\cdot z\leq p\cdot e+\sum_{j\in J}(p\cdot y(j)+ty(j)_{2})+\sum_{j\in J}t|y(j)_{2}|=p\cdot e+\sum_{j\in J}p\cdot y(j)\}. \nonumber
\]
The \emph{fuel tax demand set} is defined to be the collection of elements in $B^{c}_{\omega}(y,t,p)$ that maximizes the agent's utility function. For every $(t, p)\in \Delta$ and each $j\in J$, the firm's supply set $S^{t}_{j}(p)$ is $\argmax_{z\in Y_j}(p\cdot z+tz_{2})$. 
A $\mathcal{V}$-compliant fuel tax equilibrium under the tax rate $t$ is $(\bar{x},\bar{y},\bar{p})$ such that:
\begin{enumerate}
    \item $(t,\bar{p})\in \Delta$;
    \item $\bar{x}$ is in the fuel tax demand set;
    \item $\bar{y}(j)\in S^{t}_{j}(\bar{p})$ for all $j\in J$. Every firm is profit maximizing given the fuel tax rate $t$ and the price vector $\bar{p}$;
    \item $\bar{x}-\sum_{j\in J}\bar{y}(j)-e\in \mathcal{V}$.
\end{enumerate}

The proof of the following claim is presented in \cref{appendixexp4}.

\begin{claim}\label{fueltaxext}
There exists a $\mathcal{V}$-compliant fuel tax equilibrium for any tax rate $t\in [0, 1)$. Moreover, there exists a $\mathcal{V}$-compliant fuel tax equilibrium for any given total net $\mathrm{CO}_{2}$ emission level $v\in [0, 1]$.
\end{claim}

We now study the welfare properties of $\mathcal{V}$-compliant fuel tax equilibrium, and its comparison to $\mathcal{V}$-compliant emission tax equilibrium in \cref{taxequipbexample}. 
Coal is used to generate electricity for two purposes: consumption of electricity by the agent, and sequestration of $\mathrm{CO}_{2}$ by the second firm. Because the fuel tax applies to coal used for both purposes, it discourages sequestration and results in an inefficient outcome. We now provide the detailed calculations in the context of this example. 

As shown in \cref{taxquatunique}, for every $v\in [0, 1]$, there is a $\mathcal{V}$-compliant emission tax equilibrium such that the total net $\mathrm{CO}_{2}$ emissions is $v$. Moreover, the equilibrium consumption-production pair always takes the form $(\hat{x}_{v},\hat{y}_{v})$ where $\hat{x}_{v}=(0,0,\frac{1+v}{2})$ and $\hat{y}_{v}=\big((1,-1,1), (v-1,0,\frac{v-1}{2})\big)$. We now compare the agent's equilibrium consumption of electricity between fuel tax equilibrium and emission tax equilibrium for any given level of total net $\mathrm{CO}_{2}$ emission.  

\begin{claim}\label{eqconsumpcompare}
For a given total net $\mathrm{CO}_{2}$ emission level $v\in [0, 1]$, the agent's equilibrium consumption of electricity at every fuel tax equilibrium with total net $\mathrm{CO}_{2}$ emission $v$ is no greater than $\frac{1+v}{2}$,\footnote{As shown in \cref{taxquatunique}, $\frac{1+v}{2}$ is the agent's equilibrium consumption of electricity of the emission tax equilibrium with total net $\mathrm{CO}_{2}$ emission $v$.} with strict inequality for $v\in (0,1)$.
\end{claim}

The proof of \cref{eqconsumpcompare} is presented in \cref{appendixexp4}. 
By \cref{eqconsumpcompare}, given the the total net $\mathrm{CO}_{2}$ emission $v\in (0, 1)$, fuel tax equilibrium is Pareto dominated by the emission tax equilibrium. Given the the total net $\mathrm{CO}_{2}$ emission 
$v\in \{0, 1\}$, the corresponding emission tax equilibrium is at least as good as the corresponding fuel tax equilibrium. Finally, we note that, by the calculation in \cref{taxequipbexample}, the optimal total net $\mathrm{CO}_{2}$ emission is $\frac{1}{2}$. 
\end{customexp}

As we have seen in \cref{taxequipbexample}, emission tax equilibrium may fail to exist for certain tax rates $t$. The fuel tax formulation allows additional freedom, since the purchase price (tax plus the equilibrium price) of all commodities can be adjusted in order to equate supply and demand. However, the next example shows that fuel tax equilibrium may also fail to exist.

\begin{customexp}
{5}\label{addonnoexst}\normalfont
Let $\mathcal{F}=\{(X, \succ_{\omega}, P_{\omega}, e_{\omega}, \theta_{\omega})_{\omega\in \Omega}, (Y_j)_{j\in J}, \mathcal{V}, t, \lambda\}$ be a finite production economy with fuel tax on coal: 
\begin{enumerate}
    \item The economy $\mathcal{E}$ has three commodities $\mathrm{CO}_{2}$, coal and electricity, which we denote by $c_1$, $c_2$ and $c_3$;
    \item There are two agents with the same consumption set $X=\{0\}\times \NNReals^{2}$ and endowment $e=(0,1,0)$. Given the total net emission $v$ of $\mathrm{CO}_{2}$, the first agent's utility function $f_{v}(c_1,c_2,c_3)=c_3-\frac{v^{2}}{2}$ and the second agent's utility function is $g_{v}(c_1,c_2,c_3)=c_2+c_3-\frac{v^{2}}{2}$, i.e., the second agent derives utility from consuming both coal and electricity; 
    \item There is one  producer with production set $Y_1=\{(r,-r,r): r\in \NNReals\}$. So the producer has the production technology to burn $r$ units of coal to generate $r$ units of electricity and $r$ units of $\mathrm{CO}_{2}$ as byproduct;
    \item The compliance region $\mathcal{V}=\NegReals^{3}$;
    \item Both agents have the same shareholding of the private firm, i.e., $\theta_{1}(1)=\theta_{2}(1)=\frac{1}{2}$. The government's rebate shares to both agents are the same, i.e., $\lambda(1)=\lambda(2)=\frac{1}{2}$. 
\end{enumerate}

By the agents' utility function, the equilibrium prices for coal and electricity must both be positive. If the fuel tax is set to be greater than $\frac{1}{2}$, then the only possible equilibrium production for the producer is $(0,0,0)$. However, the first agent has a positive budget and wishes to consume electricity. Hence, there is no equilibrium. 
\end{customexp}

\section{Concluding Remarks and Future Work}\label{secconclude}
To model climate change, we introduce two GE models, the quota equilibrium model and the emission tax equilibrium model, that incorporate regulatory schemes to limit the total net pollution emissions:
\begin{enumerate}
   \item We formulate the quota equilibrium model, which includes two polar cases: cap and trade equilibrium (in which the emission property rights are vested in the private firms), and global quota equilibrium (in which the emission property rights are vested in the government, which redistributes the proceeds to the agents according to the rebate scheme). As we have seen in \cref{exampleproperty}, the property rights specified in quota equilibrium have a major impact on the distribution of welfare among agents;  
   \item \cref{standardqtmain} establishes the existence of quota equilibrium under moderate conditions. The total net emissions of the regulated commodities for any quota equilibrium equals the presepcified quota. 
   Moreover, under the assumptions of \cref{fstwelfarequota}, every quota equilibrium consumption-production pair is constrained Pareto optimal, i.e., Pareto optimal among all quota-compliant consumption-production pairs with the same total net emissions of the regulated commodities. Our examples suggest that full Pareto optimality of quota equilibrium can often be achieved if the government chooses the quota carefully;
   \item We formulate the emission tax equilibrium model in which the government allows for arbitrary net emissions of regulated commodities, but charges a tax on the regulated commodities. The government then rebates its revenue from the tax to agents according to the rebate scheme. The emission tax equilibrium model and the global quota equilibrium model allocate property rights to agents in identical ways, namely, through the government rebate scheme. In particular, as shown in \cref{equndertaxisquota} and \cref{equnderqtistax}, every quota equilibrium can be realized as an emission tax equilibrium with a corresponding tax rate and government rebate scheme, and every emission tax equilibrium can be realized as a global quota equilibrium;
   \item However, as shown in \cref{taxequipbexample} and \cref{nonlinearexample}, emission tax equilibrium need not exist for certain emission tax rates. In some cases, even if emission tax equilibria exist, it is impossible to ensure that the total net emissions of the regulated commodities fall under a prespecified level at all equilibria, as illustrated in \cref{taxequipbexample}. On the other hand, when there is a one-to-one correspondence between emission tax rate and total net pollution emission, the government can often limit the total net emissions of regulated commodities to a prespecified level and achieve full Pareto optimality, through tax alone. Finally, \cref{cbtaxequil} shows that fuel tax equilibrium may be inefficient compared with emission tax equilibrium. \cref{addonnoexst} shows that fuel tax equilibrium may not exist. 
\end{enumerate}

The paper also suggests the following promising directions for future work:
\begin{enumerate}
\item In all of our examples, full Pareto optimality of quota equilibrium can be achieved if the government sets the right quota. It would be desirable to develop general results on the existence of a quota so that the quota equilibria are full Pareto optimal;

\item It would be desirable to have a version of the Second Welfare Theorem in our setting;  

\item The Arrow-Debreu model is inherently static.\footnote{
The addition of Arrow securities allows consideration of multiple time periods in the Arrow-Debreu model, but the markets meet only once. \citet{Golosov14} analyse a dynamic stochastic general equilibrium (DSGE) model with climate change as an externality.} It would be desirable to extend our equilibrium model to  multiple time periods, in two directions:
\begin{enumerate}
\item    
Since emission-reducing technologies require long-term capital investments, the ability to trade long-term emission rights is important in facilitating the deployment of these technologies. \citet{nord07} points out that emissions price volatility is an issue for the cap and trade system, and argues that this is a reason to favor emissions taxes over cap and trade.  Volatility, and the potential for futures contracts to facilitate long-term investments in a volatile environment, can best be analyzed in a general equilibrium financial markets model with multiple time periods.  In a multi-period version of our quota model, we could allow for trading emission rights over multiple periods, as for example in emission futures; 
\item In our models, the quotas, emission taxes, and allocation schemes are exogenously specified.  Nordhaus developed optimal growth models\footnote{See \citet{RDmanual} for details of the Nordhaus DICE and RICE models.} to assist a government in {\em choosing} an optimal emission-reduction path.  
Given such an optimal path, we could readily consider a sequence of our one-period quota models, with the quota in each period prescribed by the optimal path.
It would be desirable to study the welfare properties of such a sequence of one-period equilibria, as well as of multi-period quota models in which emissions quotas can be traded across time;

\end{enumerate}

\item There is a single government in our quota and emissions tax model. Hence, our model does not allow for international trade of commodities. \citet{goulder13} point out that climate policies have a significant impact on international competitiveness of domestic firms. Thus, it would be desirable to develop quota and emissions tax models with multiple governments;  

\item In \cref{cbtaxequil}, fuel tax equilibrium is Pareto dominated by the emission tax equilibrium with the same total net $\mathrm{CO}_{2}$ emission $v\in (0, 1)$.
It would be desirable to identify sufficient conditions under which fuel tax equilibrium is generically inefficient;

\item CGE methods have been widely used in macroeconomics and applications such as international trade negotiations. Quota and tax equilibria are effectively computable using CGE methods, which can identify spillovers from regulation in the energy sector to the broader economy, including identifying losers in a policy change that is beneficial to most agents. Identifying and compensating losers are important in building a political consensus in favor of a policy change, and in mitigating the adverse side effects of that change. It would be desirable to implement CGE analysis of tax and quota equilibria in existing applied macro and trade models;

\item  In an economy with finitely many agents, all agents have some market power\footnote{\citet{roberts76} showed that the demand of each agent has some impact on the price formation as long as there are finitely many agents.} and hence may not act as price-takers. With a measure space of agents developed in \citet{au64}, all agents are negligible, and hence the price-taking assumption is justified. \citet{Noguchi2006} develop a production economy with a measure space of agents and widespread externalities, but without bads, and hence with positive prices. \citet{ha05} provides an example showing that equilibrium may not exist in a measure-theoretic production economy with bads, and in a large finite production economy, a small group of agents may end up consuming almost all of the bads.
Hence, it is desirable to develop an
extension of the production economy with quota to a measure space of agents and finitely many producers, in order to separate price-taking consumers from producers that have market power, and may need to be regulated. We propose to study the equilibrium existence and welfare properties of this extended model. 
\end{enumerate}

\appendix

\section{Appendix}\label{secappa}
The goal of this appendix is to provide a complete proof of a special case of \cref{standardqtmain} where the quota-compliance region is a cone, 
which further completes our proof of \cref{standardqtmain}. 
To do so, we need to introduce the concept of quasi-equilibrium. Let 
\[
\mathcal{E}=\{(X, \succ_{\omega}, P_\omega, e_\omega, \theta)_{\omega\in \Omega}, (Y_j)_{j\in J}, (m^{(j)})_{j\in J}, \mathcal{Z}(m)\}
\]
be a finite production economy with quota as in \cref{defmeasureeco}. 
Given $(x, y,p)\in \mathcal{A}\times Y\times \Delta$, the \emph{quota quasi-demand set} $\bar{D}^{m}_{\omega}(x,y,p)$ is defined as: 
\[
\{z\in B^{m}_{\omega}(y,p): w \succ_{x,y,\omega,p} z\implies p\cdot w\geq p\cdot e(\omega) + \sum_{j\in J}\theta_{\omega j}\big(p\cdot y(j)+\pi_{k}(p)\cdot m^{(j)}\big)\}.\nonumber
\]
A \emph{$\mathcal{Z}(m)$-compliant quota quasi-equilibrium} is $(\bar{x},\bar{y},\bar{p})\in \mathcal{A}\times Y\times \Delta$ such that: 
\begin{enumerate}
\item $\bar{x}(\omega)\in \bar{D}^{m}_{\omega}(\bar{x},\bar{y}, \bar{p})$ for all $\omega\in \Omega$;
\item $\bar{y}(j)\in S_{j}^{m}(\bar{p})$ for all $j\in J$;
\item $\sum_{\omega\in \Omega}\bar{x}(\omega)-\sum_{\omega\in \Omega}e(\omega)-\sum_{j\in J} \bar{y}(j)\in \mathcal{Z}(m)$.
\end{enumerate}
Note that quasi-equilibrium is not stable since agents could, in principle, do better within their budget sets. Thus, the interest of the quasi-equilibrium concept is purely mathematical. The rest of this appendix is broken into the following  two parts:
\begin{enumerate}
    \item To show that every quota quasi-equilibrium is a quota equilibrium under suitable regularity conditions;
    \item Prove the special case of \cref{standardqtmain} by first establishing the existence of a quota quasi-equilibrium, then applying the result mentioned in the previous item to show that the quota quasi-equilibrium is a quota equilibrium.  
\end{enumerate}

\subsection{From Quasi-Equilibrium to Equilibrium}\label{secquasieq}

In this section, we show that, under \cref{esto0} and \cref{conditionexplain} of \cref{standardqtmain}, every $\mathcal{Z}(m)$-compliant quota quasi-equilibrium is a $\mathcal{Z}(m)$-compliant quota equilibrium.  

\begin{lemma}\label{interioritylm}
Let $\mathcal{E}=\{(X, \succ_{\omega},P_{\omega}, e_{\omega}, \theta_{\omega})_{\omega\in \Omega}, (Y_j)_{j\in J}, (m^{(j)})_{j\in J}, \mathcal{Z}(m)\}$ be a finite production economy with quota and $(\bar{x},\bar{y},\bar{p})$ be a $\mathcal{Z}(m)$-compliant quota quasi-equilibrium.
Suppose
\begin{enumerate}[(i)]
\item $m^{(j)}=0$ for all $j\in J$;
\item\label{esto0appqua} there exists an agent $\omega_0\in \Omega$ such that the set $X_{\omega_0} - \sum_{j\in J}\theta_{\omega_{0} j}Y_j$ has non-empty interior $U_{\omega_0}\subset \Reals^{\ell}$ and $e(\omega_0)\in U_{\omega_0}$;
\item there exists 
a commodity $s\in \{1,2,\dotsc,\ell\}$ such that:
\begin{itemize}
    \item the projection $\pi_{s}(X_{\omega_0})$ is unbounded, and the agent $\omega_0$ has a strongly monotone preference on the commodity $s$, where $\omega_0$ is the agent specified in \cref{esto0appqua};
    \item for every $\omega\in \Omega$, there is an open set $V_{\omega}$ containing the $s$-th coordinate $e(\omega)_{s}$ of $e(\omega)$ such that $(e(\omega)_{-s},v)\in X_{\omega} - \sum_{j\in J}\theta_{\omega j}Y_j$ for all $v\in V_{\omega}$.
\end{itemize}
\end{enumerate}
Then $(\bar{x},\bar{y},\bar{p})$ is a $\mathcal{Z}(m)$-compliant quota equilibrium.
\end{lemma} 

\begin{proof}
Let $(\bar{x},\bar{y},\bar{p})$ be a $\mathcal{Z}(m)$-compliant quota quasi-equilibrium.
For each consumer $\omega$, define a correspondence  $\delta_\omega: \Delta\cto X_{\omega}$ as
\[
\delta_{\omega}(p)=\{x_{\omega}\in X_{\omega}: p\cdot x_{\omega} < p\cdot e(\omega) + \sum_{j\in J}\theta_{\omega j}\sup\{p\cdot y: y\in Y_j\}\}. \nonumber
\]
We start by establishing the following claim:
\begin{claim}\label{delta0claim}
$\delta_{\omega_0}(\bar{p})\neq \emptyset$.
\end{claim}
\begin{proof}
Note that set $X_{\omega_0} - \sum_{j\in J}\theta_{\omega_0 j}Y_j$ has non-empty interior $U_{\omega_0}$ and $e(\omega_0)\in U_{\omega_0}$. Hence, we can pick $u_{\omega_0}\in \Reals^{\ell}$ such that $\bar{p}\cdot u_{\omega_0}<0$ and that $(e(\omega_0) + u_{\omega_0}) \in (X_{\omega_0}-\sum_{j\in J}\theta_{\omega_0 j}Y_j)$. As $(\bar{x},\bar{y},\bar{p})$ is a $\mathcal{Z}(m)$-compliant quota quasi-equilibrium, we have
$\bar{p}\cdot \tilde x_{\omega_0} < \bar{p}\cdot e(\omega_0)+\sum_{j\in J}\theta_{\omega_{0} j}\bar{p}\cdot \bar{y}(j)$
for some $\tilde x_{\omega_0}\in X_{\omega_0}$.
So we have $\delta_{\omega_0}(\bar{p})\neq\emptyset$.
\end{proof}
\cref{delta0claim} leads to the following result:
\begin{claim}\label{omega0equil}
If $\hat{x}\in X_{\omega_0}$ with $(\hat{x}, \bar{x}(\omega_0))\in P_{\omega_0}(\bar{x},\bar{y},\bar{p})$, then $\bar{p}\cdot \hat{x}>\bar{p}\cdot e(\omega_0)+\sum_{j\in J}\theta_{\omega_0 j}\bar{p}\cdot \bar{y}(j)$, i.e., $\hat{x}$ is not in agent $\omega_0$'s budget set given $\bar{p}$ and $\bar{y}$. 
\end{claim}
\begin{proof}
Let $\hat{x}\in X_{\omega_0}$ be such that 
$(\hat{x}, \bar{x}(\omega_0))\in P_{\omega_0}(\bar{x},\bar{y},\bar{p})$. By \cref{delta0claim}, pick $z_{\omega_0}\in \delta_{\omega_0}(\bar{p})$. 
As $m^{(j)}=0$ for all $j\in J$, 
we have $\bar{p}\cdot \hat{x}\geq \bar{p}\cdot e(\omega_0)+\sum_{j\in J}\theta_{\omega_0 j}\bar{p}\cdot \bar{y}(j)$ since $(\bar{x},\bar{y},\bar{p})$ is a $\mathcal{Z}(m)$-compliant quota quasi-equilibrium. As $P_{\omega_0}(\bar{x},\bar{y},\bar{p})$ is continuous, there exists $\lambda\in (0,1)$ such that 
$(\lambda z_{\omega_0}+(1-\lambda)\hat{x}, \bar{x}(\omega_0)) \in P_{\omega_0}(\bar{x},\bar{y},\bar{p})$.

Assume that $\bar{p}\cdot \hat{x}=\bar{p}\cdot e(\omega_0)+\sum_{j\in J}\theta_{\omega_0 j}\bar{p}\cdot \bar{y}(j)$. 
Then we have $(\lambda z_{\omega_0}+(1-\lambda)\hat{x}, \bar{x}(\omega_0)) \in P_{\omega_0}(\bar{x},\bar{y},\bar{p})$ and $\lambda z_{\omega_0}+(1-\lambda)\hat{x}\in \delta_{\omega_0}(\bar{p})$. 
This furnishes us a contradiction since $(\bar{x},\bar{y},\bar{p})$ is a $\mathcal{Z}(m)$-compliant quota quasi-equilibrium. So we have $\bar{p}\cdot \hat{x}>\bar{p}\cdot e(\omega_0)+\sum_{j\in J}\theta_{\omega_0 j}\bar{p}\cdot \bar{y}(j)$.  
\end{proof}
As the agent $\omega_0$ has a strongly monotone preference on the commodity $s$ and the projection $\pi_{s}(X_{\omega_0})$ is unbounded, by \cref{omega0equil}, we conclude that $\bar{p}_{s}>0$, i.e., the equilibrium price $\bar{p}_{s}$ of commodity $s$ is strictly positive. 
\begin{claim}\label{deltaclaim}
For every $\omega\in \Omega$, $\delta_{\omega}(\bar{p})\neq \emptyset$. 
\end{claim}
\begin{proof}
Note that, for every $\omega\in \Omega$, there is an open set $V_{\omega}$ containing the $s$-th coordinate $e(\omega)_{s}$ of $e(\omega)$ such that $(e(\omega)_{-s},v)\in X_{\omega} - \sum_{j\in J}\theta_{\omega j}Y_j$ for all $v\in V_{\omega}$.
As $\bar{p}_{s}>0$, we can pick $u_{\omega}\in \Reals^{\ell}$ such that $\bar{p}\cdot u_{\omega}<0$ and that $(e(\omega) + u_{\omega}) \in (X_{\omega}-\sum_{j\in J}\theta_{\omega j}Y_j)$.
As $(\bar{x},\bar{y},\bar{p})$ is a $\mathcal{Z}(m)$-compliant quota quasi-equilibrium, we have
\[
\bar{p}\cdot \tilde x_{\omega} < \bar{p}\cdot e(\omega)+\sum_{j\in J}\theta_{\omega j}\bar{p}\cdot \bar{y}(j) \nonumber
\]
for some $\tilde x_\omega\in X_{\omega}$.
So we have $\delta_{\omega}(\bar{p})\neq\emptyset$.  
\end{proof}

We now show that $(\bar{x},\bar{y},\bar{p})$ is a $\mathcal{Z}(m)$-compliant quota equilibrium. The proof is similar to the proof of \cref{omega0equil}. 
For each $\omega\in \Omega$, by \cref{deltaclaim},
pick $z_{\omega}\in \delta_{\omega}(\bar{p})$ and $\hat{x}_{\omega}\in X_{\omega}$ such that $(\hat{x}_{\omega}, \bar{x}(\omega))\in P_{\omega}(\bar{x},\bar{y},\bar{p})$. 
As $m^{(j)}=0$ for all $j\in J$, 
we have $\bar{p}\cdot \hat{x}_{\omega}\geq \bar{p}\cdot e(\omega)+\sum_{j\in J}\theta_{\omega j}\bar{p}\cdot \bar{y}(j)$ since $(\bar{x},\bar{y},\bar{p})$ is a $\mathcal{Z}(m)$-compliant quota quasi-equilibrium. As $P_{\omega}(\bar{x},\bar{y},\bar{p})$ is continuous, there exists $\lambda\in (0,1)$ such that 
$(\lambda z_\omega+(1-\lambda)\hat x_\omega, \bar{x}(\omega)) \in P_{\omega}(\bar{x},\bar{y},\bar{p})$.
Assume that $\bar{p}\cdot \hat{x}_{\omega}=\bar{p}\cdot e(\omega)+\sum_{j\in J}\theta_{\omega j}\bar{p}\cdot \bar{y}(j)$. 
Then we have $(\lambda z_\omega+(1-\lambda)\hat x_\omega, \bar{x}(\omega)) \in P_{\omega}(\bar{x},\bar{y},\bar{p})$ and $\lambda z_\omega+(1-\lambda)\hat x_\omega\in \delta_{\omega}(\bar{p})$. 
This furnishes us a contradiction since $(\bar{x},\bar{y},\bar{p})$ is a $\mathcal{Z}(m)$-compliant quota quasi-equilibrium. Therefore, we have $\bar{p}\cdot \hat x_\omega>\bar{p}\cdot e(\omega)+\sum_{j\in J}\theta_{\omega j}\bar{p}\cdot \bar{y}(j)$. 
Hence, $(\bar{x},\bar{y},\bar{p})$ is a $\mathcal{Z}(m)$-compliant quota equilibrium. 
\end{proof}

\subsection{Existence of Quota Equilibrium with Quota-Compliance Cone}\label{appendixcone}

In this section, we prove a special case of \cref{standardqtmain} when the quota-compliance region is a cone. 
The main result of this section is similar to the Proposition 3.2.3 in \citet{fl03bk}.

\begin{theorem}\label{finitemainexpd}
Let $\mathcal{E}=\{(X, \succ_{\omega}, P_\omega, e_\omega, \theta_{\omega})_{\omega\in \Omega}, (Y_j)_{j\in J}, (m^{(j)})_{j\in J}, \mathcal{Z}(m)\}$ be a finite production economy with quota as in \cref{defmeasureeco}. Let the polar cone of the quota-compliance region $\mathcal{Z}(m)$ be $\mathcal{Z}^{0}=\{p\in \Delta: (\forall z\in \mathcal{Z}(m))(p\cdot z\leq 0)\}$. 
Suppose $\mathcal{E}$, in addition, satisfies: 
\begin{enumerate}[(i)]
    \item $m^{(j)}=0$ for all $j\in J$;
    \item for all $\omega\in \Omega$,  $\hat{X}_{\omega}$ is compact and $P_{\omega}$ takes value in $\mathcal{P}_{H}$; 
    \item \label{esto0app} there exists an agent $\omega_0\in \Omega$ such that the set $X_{\omega_0} - \sum_{j\in J}\theta_{\omega_{0} j}Y_j$ has non-empty interior $U_{\omega_0}\subset \Reals^{\ell}$ and $e(\omega_0)\in U_{\omega_0}$;
\item \label{estcommoapp} there exists 
a commodity $s\in \{1,2,\dotsc,\ell\}$ such that:
\begin{itemize}
    \item the projection $\pi_{s}(X_{\omega_0})$ is unbounded, and the agent $\omega_0$ has a strongly monotone preference on the commodity $s$, where $\omega_0$ is the agent specified in \cref{esto0app};
    \item for every $\omega\in \Omega$, there is an open set $V_{\omega}$ containing the $s$-th coordinate $e(\omega)_{s}$ of $e(\omega)$ such that $(e(\omega)_{-s},v)\in X_{\omega} - \sum_{j\in J}\theta_{\omega j}Y_j$ for all $v\in V_{\omega}$;
\end{itemize}
    \item \label{satiateappen}for all $\omega\in \Omega$, for each $(x,y)\in \cO$, there exists $u\in X_\omega$ such that $(u, x_\omega)\in \bigcap_{p\in \Delta\cap \mathcal{Z}^{0}}P_{\omega}(x,y,p)$;
    \item the aggregate production set $\bar{Y}=\left\{\sum_{j\in J}{y}(j): y\in Y\right\}$ is closed and convex, and for each $j\in J$, the set $\hat{Y}_{j}$ is relatively compact;
    \item $\bar{Y}+\mathcal{Z}(m)$ is closed.
\end{enumerate}
Then, there exists $(\bar{x}, \bar{y}, \bar{p})\in \mathcal{A}\times Y\times \Delta$ such that:
\begin{enumerate}
    \item $(\bar{x}, \bar{y}, \bar{p})$ is a $\mathcal{Z}(m)$-compliant quota equilibrium;
    \item we have $\bar{p}\in \mathcal{Z}^{0}$ and $\bar{p}\cdot \big(\sum_{\omega\in \Omega}\bar{x}(\omega)-\sum_{\omega\in \Omega}e(\omega)-\sum_{j\in J}\bar{y}(j)\big)=0$.
\end{enumerate}
\end{theorem}
\begin{proof}
As $m=\sum_{j\in J}m^{(j)}=0$, the quota-compliance region $\mathcal{Z}(m)$ is a cone, and neither the government firm nor the private firms have any quota. Thus, \cref{finitemainexpd} is similar to Proposition 3.2.3 in \citet{fl03bk}. 
For $\omega\in \Omega$, define the correspondence $P'_{\omega}: \cA\times Y\times \Delta\cto X_\omega$ by 
\[
P'_{\omega}(x,y,p)=\{a\in X_\omega|(a,x_\omega)\in P_{\omega}(x,y,p)\}.\nonumber
\]
Note that $P'_{\omega}$ is lower hemicontinuous since the preference map $P_{\omega}$ is continuous. As $P_{\omega}$ takes value in $\mathcal{P}_{H}$ and is irreflexive, $x_{\omega}\not\in \conv(P'_{\omega}(x,y,p))$ for all
$(x,y,p)\in \cA\times Y\times \Delta$ and all $\omega\in \Omega$. By \cref{satiateappen}, we have $\bigcap_{p\in \Delta\cap \mathcal{Z}^{0}}P'_{\omega}(x,y,p)\neq \emptyset$ for all $(x,y)\in \cO$. By \cref{esto0app} and \cref{estcommoapp}, we have $e(\omega)\in X_{\omega}-\sum_{j\in J}\theta_{\omega j}Y_j$ for all $\omega\in \Omega$.
As $\hat{X}_{\omega}$ is compact for every $\omega\in \Omega$, $\bar{Y}$ is closed and convex, $\hat{Y}_{j}$ is relatively compact for every $j\in J$ and $\bar{Y}+\mathcal{Z}(m)$ is closed, 
by Proposition 3.2.3 in \citet{fl03bk}, we conclude that $\mathcal{E}$ has a $\mathcal{Z}(m)$-compliant quota quasi-equilibrium $(\bar{x}, \bar{y}, \bar{p})\in \cA\times Y\times \Delta$. Moreover, we have $\bar{p}\in \mathcal{Z}^{0}$ and $\bar{p}\cdot \big(\sum_{\omega\in \Omega}\bar{x}(\omega)-\sum_{\omega\in \Omega}e(\omega)-\sum_{j\in J}\bar{y}(j)\big)=0$. 
As $m^{(j)}=0$ for all $j\in J$, 
by \cref{interioritylm}, $(\bar{x}, \bar{y}, \bar{p})$ is a $\mathcal{Z}(m)$-compliant quota equilibrium.
\end{proof}

\section{Detailed Analysis of Examples}\label{appendixexp}
This appendix is devoted to detailed analysis of the examples in the main body of the paper. In particular, we provide rigorous proofs to the claims within the examples. 
\subsection{Detailed Analysis of \cref{exampleproperty}}\label{appendixexp1}
We provide rigorous proofs for \cref{governquota} and \cref{privatequota} to complete our analysis of \cref{exampleproperty}.

\begin{proof}[Proof of \cref{governquota}]
As $\mathrm{CO}_{2}$ comes solely from the production of the private firm, the private firm's equilibrium production is $(-m,m,-m)$. Note that $\bar{p}_{1}-\bar{p}_{2}+\bar{p}_{3}=0$ since otherwise the first firm's profit is unbounded. As both agents have the same endowment, utility function and share of the government firm, and we also require no-free-disposal of electricity at equilibrium, they must each consume $\frac{-m}{2}$ unit of electricity at any equilibrium. 
Hence, we have $\frac{-m}{2}\cdot \bar{p}_{3}=\bar{p}_{2}-\frac{-m}{2}\bar{p}_{1}$, which implies that $\bar{p}_{3}=\frac{2}{-m}\bar{p}_{2}-\bar{p}_{1}$. As $\bar{p}_{1}-\bar{p}_{2}+\bar{p}_{3}=0$ and $m\in (-2,0)$, we conclude that $\bar{p}_{2}=0$. 
By the form of agents' utility function, we have $\bar{p}_{3}>0$. As $\bar{p}\in \Delta$, we conclude that $\bar{p}=(-\frac{1}{2},0,\frac{1}{2})$.
\end{proof}

\begin{proof}[Proof of \cref{privatequota}]
As in the proof of \cref{governquota}, the equilibrium production for the private firm is $(-m,m,-m)$, and we have $\bar{p}_{1}-\bar{p}_{2}+\bar{p}_{3}=0$. The first agent's budget is $\bar{p}_{2}+m\bar{p}_{1}=-m(\bar{p}_{2}-\bar{p}_{1})+(1+m)\bar{p}_{2}$ since the first agent's shareholding of the private firm is $1$. Again, by the form of agents' utility function, we have $\bar{p}_{3}>0$.
Intuitively, since $m>-2$ and no agent derives utility from coal, the price of coal must be $0$. 
We now rigorously establish that $\bar{p}_{2}=0$ by considering three different cases $m<-1$, $m=-1$ and $m>-1$: 
\begin{itemize}
    \item If $m<-1$, as $\bar{p}_{3}>0$, then $\bar{p}_{2}\geq 0$ since the first agent's equilibrium consumption of electricity can not exceed $-m$ units. The second agent's budget is $\bar{p}_{2}$ which implies that the aggregate equilibrium consumption of electricity is $\frac{2\bar{p}_{2}+m\bar{p}_{1}}{\bar{p}_{3}}=\frac{-m(\bar{p}_{2}-\bar{p}_{1})+(2+m)\bar{p}_{2}}{\bar{p}_{3}}$. As $m>-2$, the aggregate equilibrium consumption of electricity exceeds $-m$ if $\bar{p}_{2}>0$. Hence, we have $\bar{p}_{2}=0$;
    \item If $m=-1$, then the first agent's budget is $\bar{p}_{2}-\bar{p}_{1}$. Thus, the first agent consumes $1$ unit of electricity which implies that the second agent's equilibrium electricity consumption must be $0$. As the second agent's budget is given by $\bar{p}_{2}$ and $\bar{p}_{3}>0$, we conclude that $\bar{p}_{2}=0$;
    \item If $m>-1$, as $\bar{p}_{3}>0$, then $\bar{p}_{2}\leq 0$ since the first agent's equilibrium consumption of electricity can not exceed $-m$ units. Note that the second agent's budget is $\bar{p}_{2}$. As $\bar{p}_{3}>0$ and $X_{2}\subset \NNReals^{3}$, we have $\bar{p}_{2}\geq 0$. Thus, we conclude that $\bar{p}_{2}=0$.
\end{itemize}
As $\bar{p}\in \Delta$, $\bar{p}_{1}-\bar{p}_{2}+\bar{p}_{3}=0$ and $\bar{p}_{3}>0$, we conclude that $\bar{p}=(-\frac{1}{2},0,\frac{1}{2})$. Hence, the first agent's equilibrium electricity consumption is $-m$ while the second agent's equilibrium electricity consumption is $0$.
\end{proof}

\subsection{Detailed Analysis of \cref{taxequipbexample}}\label{appendixexp2} 
We provide rigorous proofs for \cref{eqestclaim}, \cref{wrongtax} and \cref{taxquatunique} to complete our analysis of \cref{taxequipbexample}. 

\begin{proof}[Proof of \cref{eqestclaim}]
Let $t_0\leq \frac{1}{4}$ be the tax rate on $\mathrm{CO}_{2}$. Let $\bar{x}=(0,0,1), \bar{y}=\big((1,-1,1),(0,0,0)\big)$ and $\bar{p}=(-t_0, \frac{1}{2}-t_0, \frac{1}{2})$. We claim that $(\bar{x},\bar{y},\bar{p})$ is a $\mathcal{V}$-compliant emission tax equilibrium with tax rate $t_0$ on $\mathrm{CO}_{2}$. At the equilibrium price $\bar{p}$, both firms are profit maximizing, which are both $0$. The budget set for the agent is:
\[
\{z\in X: \bar{p}\cdot z\leq \frac{1}{2}-t_0+t_0=\frac{1}{2}\}.
\]
Hence, we have $\bar{x}\in D^{t}(\bar{x},\bar{y},\bar{p})$. Note that $\bar{x}-e-\sum_{j\in J}\bar{y}(j)=(-1,0,0)\in \mathcal{V}$. So $(\bar{x},\bar{y},\bar{p})$ is a $\mathcal{V}$-compliant emission tax equilibrium with tax rate $t_0$. 

We now show that there is no $\mathcal{V}$-compliant emission tax equilibrium if the emission tax rate is greater than $\frac{1}{4}$. Suppose $t>0$ is an emission tax rate on $\mathrm{CO}_{2}$ under which there is a $\mathcal{V}$-compliant emission tax equilibrium $(\hat{x},\hat{y},\hat{p})$. By definition, we know that $\hat{p}_{1}=-t$. The equilibrium price $\hat{p}_{3}$ must be no less than $2t$ since otherwise the second firm's profit is unbounded. 
For the same reason, we know that $\hat{p}_{2}\geq \hat{p}_{3}-t>0$.
As the endowment $e=(0,1,0)$, the agent's budget at equilibrium is positive. The equilibrium production for the first firm must not be $(0,0,0)$ since the agent has a positive budget which she will spend entirely on electricity. 
Hence, we conclude that $\hat{p}_{3}=t+\hat{p}_{2}$. As $\hat{p}\in \Delta$, we have $2t+2\hat{p}_{2}=1$, which implies that $\hat{p}_{3}=\frac{1}{2}$. As $\hat{p}_{3}\geq 2t$, we know that $t\leq \frac{1}{4}$.
By \cref{eqestclaim}, we conclude that $\mathcal{F}$ has a $\mathcal{V}$-compliant emission tax equilibrium if and only if the emission tax rate $t\leq \frac{1}{4}$. Hence, we conclude that there exists a $\mathcal{V}$-compliant emission tax equilibrium if and only if the tax rate $t\leq \frac{1}{4}$. 
\end{proof}

\begin{proof}[Proof of \cref{wrongtax}]
Pick $t_0<\frac{1}{4}$ to be the emission tax rate of $\mathrm{CO}_{2}$. As indicated in \cref{eqestclaim}, there is a $\mathcal{V}$-compliant emission tax equilibrium $(\hat{x},\hat{y},\hat{p})$. By the same argument as in the previous paragraph, we conclude that $\hat{p}_{3}=\frac{1}{2}$ and $\hat{p}_{2}=\frac{1}{2}-t_0$.
As $\hat{p}_{3}>2t_0$, the equilibrium production for the second firm is $(0,0,0)$.
Suppose the equilibrium production for the first firm is $(r,-r,r)$ for $r<1$. The emission tax budget set is:
\[
\{z\in X: \hat{p}\cdot z\leq \frac{1}{2}-t_0+rt_0\}. 
\]
As $t_0<\frac{1}{4}$ and $r<1$, we have $\frac{1}{2}r<\frac{1}{2}-(1-r)t_0$, which implies that the consumption $(0,0,r)$ is in the agent's budget set. 
However, the agent has extra budget to consume more electricity but the total production of the electricity is $r$ units. So $(0,0,r)$ is not in the emission tax demand set $D^{t}(\hat{x},\hat{y},\hat{p})$.
Hence, the equilibrium production for the first firm is $(1,-1,1)$. So the total net emission of $\mathrm{CO}_{2}$ is $1$.
\end{proof}

\begin{proof}[Proof of \cref{taxquatunique}]
Pick some $v_0\in [0,1]$. 
To achieve this pre-specified total net emission of $\mathrm{CO}_{2}$, all possible equilibrium production plans take the following form:
\begin{itemize}
    \item The first firm produces at $(r_1,-r_1, r_1)$ where $0\leq r_1\leq 1$. The second firm produces at $(-2r_2,0,-r_2)$ where $0\leq r_2\leq \frac{r_1}{2}$. We also require $r_1-2r_2=v_0$.
\end{itemize}
Thus, the agent's utility function is $(r_1-r_2)-\frac{v_{0}^{2}}{2}$. As $r_2=\frac{r_{1}-v_{0}}{2}$, so the agent's utility maximizes by taking $r_{1}=1$. By \cref{equnderqtistax} and \cref{fstwelfarequota}\footnote{Note that we allow for free disposal of electricity and coal at equilibrium. However, as the equilibrium price for electricity and coal are both positive, any equilibrium consumption-production pair must be constrained Pareto optimal.}, the only possible equilibrium production plan for the first firm is $(1,-1,1)$, which implies that the only possible equilibrium consumption-production pair is $\hat{x}_{v_0}=(0,0,\frac{1+v_0}{2})$ and $\hat{y}_{v_0}=\big((1, -1, 1), (v_0-1, 0, \frac{v_0-1}{2})\big)$. It remains to show that $(\hat{x}_{v_0},\hat{y}_{v_0},\hat{p})$ is a $\mathcal{V}$-compliant emission tax equilibrium. 
Note that $\hat{x}_{v_0}-e-\sum_{j\in J}\hat{y}_{v_0}(j)=(-v_0,0,0)\in \mathcal{V}$. Both firms are profit maximizing. The emission tax budget set for the agent is
$\{z\in X: \hat{p}\cdot z\leq \frac{1}{4}(1+v_0)\}$. Hence, $\hat{x}_{v_0}$ is an element of the emission tax demand set $D^{t}(\hat{x}_{v_0},\hat{y}_{v_0},\hat{p})$, which implies that $(\hat{x}_{v_0}, \hat{y}_{v_0}, \hat{p})$ is a $\mathcal{V}$-compliant emission tax equilibrium for the emission tax rate $\frac{1}{4}$. 
\end{proof}

\subsection{Detailed Analysis of \cref{nonlinearexample}}\label{appendixexp3}
We provide rigorous proofs for \cref{uniqueequilclaim} to complete the analysis of \cref{nonlinearexample}.

\begin{proof}[Proof of \cref{uniqueequilclaim}]
Pick $t_0\leq \frac{1}{4}$. Let $\bar{x}_{t_0}=(0,0,1-4t_{0}^{2})$, $\bar{y}_{t_0}=\big((1,-1,1),(-4t_{0},0,-4t_{0}^{2})\big)$ and $\bar{p}_{t_0}=(-t_{0},\frac{1}{2}-t_{0},\frac{1}{2})$. We claim that $(\bar{x}_{t_0},\bar{y}_{t_0},\bar{p}_{t_0})$ is a $\mathcal{V}$-compliant emission tax equilibrium with emission tax rate $t_0$ on $\mathrm{CO}_{2}$. It is clear that the first firm is profit maximizing and $\bar{x}_{t_0}-e-\sum_{j\in J}\bar{y}_{t_0}(j)=(4t_{0}-1,0,0)\in \mathcal{V}$. The second firm's profit at $\bar{p}_{t_0}$, as a function of production, is given by $2rt_0-\frac{1}{2}r^{2}$. Thus, the second firm maximize its profit at $(-4t_{0},0,-4t_{0}^{2})$, and its profit is $2t_{0}^{2}$. The agent's emission tax budget set is:
\[
\{z\in X: \bar{p}\cdot z\leq \frac{1}{2}-t_{0}+2t_{0}^{2}+(1-4t_0)t_0=\frac{1}{2}-2t_{0}^{2})\}. \nonumber
\]
Hence, $\bar{x}_{t_0}$ is an element of the emission tax demand set $D^{t_0}(\bar{x}_{t_0},\bar{y}_{t_0},\bar{p}_{t_0})$, which implies that $(\bar{x}_{t_0},\bar{y}_{t_0},\bar{p}_{t_0})$ is a $\mathcal{V}$-compliant emission tax equilibrium with emission tax rate $t_0$. 

We now show that $(\bar{x}_{t_0},\bar{y}_{t_0},\bar{p}_{t_0})$ is the only $\mathcal{V}$-compliant emission tax equilibrium with emission tax rate $t_0$. 
Suppose $(\hat{x},\hat{y},\hat{p})$ is a $\mathcal{V}$-compliant emission tax equilibrium with emission tax rate $t_0$. By the form of the agent's utility function, we know that $\hat{p}_{3}>0$. Note that the second firm's profit at $\hat{p}$, as a function of production, is given by $2rt_{0}-r^{2}\hat{p}_{3}$. So the second firm maximize its profit by producing at $(-\frac{2t_0}{\hat{p}_{3}},0, -(\frac{t_0}{\hat{p}_{3}})^{2})$. As $(\hat{x},\hat{y},\hat{p})$ is a $\mathcal{V}$-compliant emission tax equilibrium and the total emission of $\mathrm{CO}_{2}$ can not exceed $1$ unit, so we have $\frac{2t_0}{\hat{p}_{3}}\leq 1$, which implies that $\hat{p}_{3}\geq 2t_0$. Note that we must have $t_0+\hat{p}_{2}\geq \hat{p}_{3}$ since otherwise the first firm's profit is unbounded. Thus, we conclude that $\hat{p}_{2}\geq t_0$ so the agent's endowment is positive, which further implies that the agent's budget at equilibrium is positive. By the form of the agent's utility function, the agent spends all its budget to consume electricity hence the equilibrium production for the first firm must not be $(0,0,0)$. Thus, we must have $\hat{p}_{3}=\hat{p}_{2}+t_0$. Since $\hat{p}\in \Delta$, we have $\hat{p}=(-t_0,\frac{1}{2}-t_0,\frac{1}{2})$. So the second firm's equilibrium production is $(-4t_0,0,-4t_0^{2})$. Suppose the first firm's equilibrium production is $(r,-r,r)$. Then the agent's emission tax budget set is:
\[\label{emitaxbudget}
\{z\in X: \hat{p}\cdot z\leq \frac{1}{2}-t_0+2t_{0}^{2}+(r-4t_{0})t_{0}=\frac{1}{2}-(1-r)t_0-2t_0^{2}\}. \nonumber
\]
So the emission tax demand set $D^{t_0}(\hat{x},\hat{y},\hat{p})$ is $\{(0,0,1-2(1-r)t_0-4t_{0}^{2})\}$. As $\hat{y}_{1}+\hat{y}_{2}=(r-4t_0, -r, r-4t_0^{2})$, so there are $r-4t_0^{2}$ unit of electricity in the economy available to the agent. 
If $r<1$, then we have $r-4t_0^{2}<1-2(1-r)t_0-4t_0^{2}$. So the agent has enough budget to consume more electricity than what is available to her. As a result, the first firm's equilibrium production must be $(1,-1,1)$, which implies that $\hat{x}=(0,0,1-4t_0^{2})$. Hence, $(\bar{x}_{t_0},\bar{y}_{t_0},\bar{p}_{t_0})$ is the unique $\mathcal{V}$-compliant tax equilibrium with the emission tax rate $t_0$. 

We now show that there is no $\mathcal{V}$-compliant emission tax equilibrium if the emission tax rate is greater than $\frac{1}{4}$. Suppose $t>0$ is an emission tax rate on $\mathrm{CO}_{2}$ under which there is a $\mathcal{V}$-compliant emission tax equilibrium $(\bar{x}_{t},\bar{y}_{t},\bar{p}_{t})$. By the same argument as in \cref{uniqueequilclaim}, the second firm's equilibrium production is $(-4t_0,0,-4t_0^{2})$. As the total emission of $\mathrm{CO}_{2}$ can not exceed $1$ unit, this implies that $t_0\leq \frac{1}{4}$.
So there is no $\mathcal{V}$-compliant emission tax equilibrium if the tax rate is greater than $\frac{1}{4}$. 
\end{proof}

\subsection{Detailed Analysis of \cref{cbtaxequil}}\label{appendixexp4}
We provide rigorous proofs for \cref{fueltaxext} and \cref{eqconsumpcompare} to complete the analysis of \cref{cbtaxequil}. 

\begin{proof}[Proof of \cref{fueltaxext}]
We first show that there exists a $\mathcal{V}$-compliant fuel tax equilibrium for any tax rate $t\in [0, 1)$. We break our analysis into two cases:
\begin{enumerate}
\item Suppose $t<\frac{1}{2}$. Let $(x,y,p)$ be such that $x=(0,0,1)$, $y=\big((1,-1,1),(0,0,0)\big)$ and $p=(0,\frac{1}{2}-t,\frac{1}{2})$. It is easy to see that $(x,y,p)$ is a $\mathcal{V}$-compliant fuel tax equilibrium;
\item Suppose $t\geq \frac{1}{2}$. Let $(x,y,p)$ be such that $x=(0,0,0)$, $y=\big((0,0,0),(0,0,0)\big)$ and $p=(0,0,1-t)$. It is easy to see that $(x,y,p)$ is a $\mathcal{V}$-compliant fuel tax equilibrium.
\end{enumerate}

We now show that there exists a $\mathcal{V}$-compliant fuel tax equilibrium for any given total net $\mathrm{CO}_{2}$ emission level $v\in [0, 1]$. 
Given $v\in [0, 1]$, let $(x_{v},y_{v},p_{v})$ be such that $x_{v}=(0,0,v)$, $y_{v}=\big((v,-v,v),(0,0,0)\big)$ and $p_{v}=(0,0,\frac{1}{2})$. It is straightforward to verify that $(x_{v},y_{v},p_{v})$ is a $\mathcal{V}$-compliant fuel tax equilibrium such that the total net $\mathrm{CO}_{2}$ emission is $v$.
\end{proof}

\begin{proof}[Proof of \cref{eqconsumpcompare}]
Given a total net $\mathrm{CO}_{2}$ emission level $v$, let $(\bar{x}_{v},\bar{y}_{v},\bar{p}_{v})$ be a $\mathcal{V}$-compliant fuel tax equilibrium such that the total net $\mathrm{CO}_{2}$ emission is $v$. We break our analysis into the following two cases:
\begin{enumerate}
\item We first consider the case where $\bar{p}_{v}(1)\geq 0$. By the agent's utility function, we have $\bar{p}_{v}(3)>0$. Hence, the equilibrium production for the second firm is $(0,0,0)$. As the total net $\mathrm{CO}_{2}$ emission is $v$, the first firm's equilibrium production must be $(v,-v,v)$. Hence, the agent's equilibrium consumption of electricity is bounded by $v$. Recall that the agent's equilibrium electricity consumption is $\frac{1+v}{2}$ in the corresponding emission tax equilibrium. It is clear that $v\leq \frac{1+v}{2}$ for all $v\in [0, 1]$, and the equality holds only when $v=1$;

\item We then consider the case where $\bar{p}_{v}(1)<0$. In particular, we focus on the case where the second firm's equilibrium production is not $(0,0,0)$ since otherwise it is the same as the above case. As the second firm's equilibrium production is not $(0,0,0)$, we conclude that $\bar{p}_{v}(3)=-2\bar{p}_{v}(1)$. Note that the first firm's equilibrium production must also not be $(0,0,0)$ hence we have $\bar{p}_{v}(2)+t=\bar{p}_{v}(3)+\bar{p}_{v}(1)=-\bar{p}_{v}(1)$. Suppose the first firm's equilibrium production is $(r,-r,r)$. Then the agent's budget is given by $\bar{p}_{v}(2)+r\big(\bar{p}_{v}(1)+\bar{p}_{v}(3)-\bar{p}_{v}(2)\big)$. By the form of the agent's utility function, the agent's equilibrium electricity consumption is given by $r-\big(\frac{r}{2}-\frac{(1-r)\bar{p}_{v}(2)}{\bar{p}_{v}(3)}\big)$.  The second firm's equilibrium production which maximizes the agent's utility is given by $(r-\frac{2(1-r)\bar{p}_{v}(2)}{\bar{p}_{v}(3)}, 0, \frac{r}{2}-\frac{(1-r)\bar{p}_{v}(2)}{\bar{p}_{v}(3)})$.\footnote{The agent's utility is maximized when the second firm uses all the remaining electricity to sequester $\mathrm{CO}_{2}$.} As a result, the total net $\mathrm{CO}_{2}$ emission is given by $\frac{2(1-r)\bar{p}_{v}(2)}{\bar{p}_{v}(3)}$. Hence, if we are given the total net $\mathrm{CO}_{2}$ emission $v$ (that is, $\frac{2(1-r)\bar{p}_{v}(2)}{\bar{p}_{v}(3)}=v$), the agent's equilibrium consumption is $\frac{r+v}{2}$. It is clear that $\frac{r+v}{2}\leq \frac{1+v}{2}$ for all $r\in [0, 1]$, and the equality holds only when $r=1$ which leads to $v=0$. 
\end{enumerate} 
Combining the two cases together, we obtain the desired result. 
\end{proof}

\section{Bibliography}
\printbibliography[heading=none]

\end{document}